%% file: mis-specified.tex
\title{Misspecified Beliefs about Time Lags\thanks{We thank S. Nageeb Ali, Renee Bowen, Drew Fudenberg, Yuhta Ishii, Shengwu Li, and Bruno Strulovici for helpful comments. We thank NSF Grant SES-1947021 for financial support.}}
\author{Yingkai Li \and Harry Pei}
\date{\today}

\documentclass[11pt]{article}
\usepackage{amsmath}
\usepackage{graphicx}
\usepackage{amsfonts}
\usepackage{amsthm}
\usepackage{setspace}
\usepackage{tikz}
\usepackage{amssymb}
\usetikzlibrary{patterns}
\usepackage{sgame}
\usepackage{color}
\usepackage{hyperref}
\usepackage{times}
\usepackage{enumitem}
\usepackage{mathptmx}
\usepackage[bottom]{footmisc}
\usepackage[top=1in, bottom=1in, left=1in, right=1in]{geometry}
\usepackage{cleveref}
\usepackage{natbib}
\setcitestyle{authoryear,open={(},close={)}}
\include{math}
\begin{document}

\clearpage\maketitle
\thispagestyle{empty}

\begin{abstract} \normalsize
We examine the long-term behavior of a Bayesian agent who has a misspecified belief about the time lag between actions and feedback, and learns about the payoff consequences of his actions over time. 
Misspecified beliefs about time lags result in attribution errors, which have no long-term effect when the agent's action converges, but can lead to arbitrarily large long-term inefficiencies when his action cycles. Our proof uses concentration inequalities to bound the frequency of action switches, which are useful to study learning problems with history dependence. We apply our methods to study a policy choice game between a policy-maker who has a correctly specified belief about the time lag and 
the public who has a  misspecified belief.\\

\end{abstract}

\numberwithin{equation}{section}

\begin{spacing}{1.5}
\newtheorem{Proposition}{\hskip\parindent\bf{Proposition}}
\newtheorem{Theorem}{\hskip\parindent\bf{Theorem}}
\newtheorem{Lemma}{\hskip\parindent\bf{Lemma}}[section]
\newtheorem{Corollary}{\hskip\parindent\bf{Corollary}}[section]
\newtheorem*{Condition}{\hskip\parindent\bf{Regular Prior Belief}}
\newtheorem*{Definition2}{\hskip\parindent\bf{Valid Cycle}}
\newtheorem*{Definition3}{\hskip\parindent\bf{Concentration Inequality with Unbounded Support}}
\newtheorem{Definition}{\hskip\parindent\bf{Definition}}
\newtheorem{Assumption}{\hskip\parindent\bf{Assumption}}
\newtheorem{Claim}{\hskip\parindent\bf{Claim}}

\newpage

\setcounter{page}{1}

\section{Introduction}\label{sec1}
We study learning problems faced by Bayesian decision makers who have misspecified beliefs about the time lag between decisions and feedback. We examine the long-term consequences of such belief misspecifications both in single-agent decision-making problems and in games of collective decision-making.

Misperception about time lags is prevalent among decision makers at various levels, ranging from leaders in organizations to ordinary citizens. For example, a manager decides how much resource to allocate to  R\&D. Unlike efforts on production and sales, investments in R\&D are unlikely to pay off in the short run, and moreover, it is usually unclear when and whether they will pay off.  \cite{RS-02}  show that these time lags hinder an organization's learning about the optimal resource allocation by ``\textit{complicating the attribution of causality between actions and results}''.
\citet*{RRS-09} point out that what slows down organizational learning is not the delay per se, but instead, people's misperceptions about the delay. Consequences of such misperceptions include the so-called \textit{capability traps} \citep{RS-02}, in which members of an organization work hard on production at the expense of cutting back on the time allocated to R\&D and maintenance, that ultimately results in low productivity.

Similarly, fans of football clubs tend to credit or blame their current managers for their team's performances
while ignoring the effects of previous managers' decisions. Many people believe that reopening the economy is safe amidst the COVID-19 pandemic when the number of cases and hospitalizations in Georgia, Florida, and Arizona went down three weeks after these states' reopenings.\footnote{The state of Georgia reopened in late April, and on May 23rd, Governor Brian Kemp shared the news that hospitalizations are down by 30\% since the state reopened. Similar patterns arise after the reopening of Florida, Texas and Arizona in early May. However, the number of cases and hospitalizations in these states started to surge from late June to July. See  https://www.latimes.com/world-nation/story/2020-05-23/georgia-reopened-first-the-data-say-whatever-you-want-them-to and  https://www.cnbc.com/2020/06/29/more-states-reverse-or-slow-reopening-plans-as-coronavirus-cases-climb.html} However, the number of cases and hospitalizations started to surge six to eight weeks after these states' reopennings, forcing some of them to partially return to lockdown.

We propose a model that incorporates such misperceptions. 
In every period, 
an agent chooses an action and observes an outcome that determines his payoff. The agent faces uncertainty about the \textit{state}, i.e., the mapping from his actions to
the outcome distributions. 
He observes the history of actions and outcomes and updates his belief according to Bayes rule. We assume that the true state belongs to the support of the agent's prior belief and that the outcome is informative about the state regardless of the agent's action.

The outcome distribution in period $t$ depends only on the agent's action in period $t-k^*$ while the agent believes that it depends on his action in period $t-k'$, where $k'$ is different from $k^*$.\footnote{Section \ref{sec5} extends our  result to situations where (1) the outcome distribution depends on a weighted average of the agent's current and past actions, or (2) the agent faces uncertainty about the time lag and learns about it over time.}
Our formulation can capture, for example, 
an individual underestimates the time it takes for workouts to have effects on fitness, 
a policy-maker underestimates or overestimates the time it takes for a curriculum reform to have effects on students' academic achievements, and so on.

This novel form of belief misspecification interferes learning through an \textit{attribution error}, which has no long-term effect when the agent's action converges but can lead to mislearning when the agent's action changes over time. 
Theorem \ref{Theorem1} shows that the mislearning caused by attribution errors can lead to arbitrarily large long-term inefficiencies 
in the sense that for an open set of states, there exist  prior beliefs that include the true state in their support such that the asymptotic frequency with which the agent takes his optimal action is arbitrarily close to zero. 
This stands in contrast to the benchmark scenario with a correctly  specified belief about the time lag, in which the agent chooses his optimal action almost surely in the long run.

The first challenge in establishing  this result stems from the fact that our learning problem exhibits nontrivial history-dependence. This is because the agent's current-period action directly affects his future observations. 
The second challenge arises from the observation that the
agent's action cannot converge to anything suboptimal and inefficiencies can only arise when the agent's action cycles in the long run. As a result, one needs to  bound the frequency of action switches in order to quantify the amount of mislearning, which is a key step toward showing that the posterior probability of the true state is low in the long run.

We develop a new technique using concentration inequalities. First, we examine an auxiliary problem in which the true state is excluded from the agent's prior belief. We use the Chernoff-Hoeffding inequality to show that in expectation, the agent switches actions within a finite number of periods. We then establish a concentration inequality on unbounded random variables
in order to bound the frequency of action switches. 
Next, we study situations in which the true state occurs with small but positive probability. We use the Azuma-Hoeffding inequality to show that due to the mislearning caused by frequent action switches, 
the true state occurs with low probability in the agent's posterior for all periods. This explains why the agent's actions cycle over time even when the true state belongs to the support of his prior belief.

We apply our framework to study a dynamic policy choice game between a policy-maker who has a correctly specified belief about the time lag and the public who has a misspecified belief. The policy-maker wants to implement a socially beneficial reform but cannot do so without the public's support.  The public prefers the reform to the status quo in one state and prefers the status quo in the other state. This conflict of interest can arise 
when the reform has positive externalities or the status quo has negative externalities on marginalized groups 
 (e.g., massive gatherings during a pandemic has negative externalities on the immunocompromised)
which the policy-maker cares about but the majority of citizens fail to internalize. Therefore, a reform can be both optimal for the benevolent policy-maker and suboptimal for the majority of citizens.

We characterize the maximal frequency that the policy-maker can implement the reform when he has no private information about the state. We show that the policy-maker's optimal payoff equals the maximal frequency of reform in an auxiliary game where he knows the state but the public is naive in the sense that they fail to recognize the informational content of the policy-maker's behaviors. Intuitively, this is because the policy-maker can asymptotically learn the true state, and when the reform is optimal for the public, he can implement the reform in almost every period regardless of the public's prior.

We also construct a class of strategies under which the policy-maker can approximately achieve his optimal payoff, according to which he proposes the reform with frequency close to a half when the public entertains a pessimistic belief about the reform, and proposes the reform with frequency strictly greater than a half when the public entertains an optimistic belief about the reform. The former maximizes the amount of mislearning and the latter maximizes the frequency of  reform subject to a constraint that the expected amount of mislearning is non-negative. The key step is to use
the Wald inequality and show that conditional on the reform being suboptimal for the public,
the policy-maker's future proposals are accepted with probability close to $1$
when he started to propose the reform with frequency greater than a half.

 Our work contributes to a growing literature on misspecified learning by studying  environments with history dependence. The agent in our model has a misspecified belief about the dynamic structure of the problem and his past actions can affect  future outcomes. This stands in contrast to most of the existing works such as \cite{Ber-66}, \cite{Nya-19}, \cite{EP-16}, \citet*{FRS-17}, \cite{BH-20}, \citet*{FII-20}, \citet*{EPY2020}, and \citet*{FLS-20} that exclude history dependence.

Several recent papers study misspecifed learning models with history dependence and provide conditions for the steady states. \cite*{shalizi2009dynamics} provides sufficient conditions for the convergence of posterior belief when there is no endogenous action choice and the signals in different periods can be correlated. 
\cite{He-20} examines misspecified learning in two-period optimal stopping problems in which an agent mistakenly believes that the outcome in the second period is negatively correlated with that in the first period. \cite{EP-20} study a single-agent Markov decision problem with misspecified beliefs about the state transition function. \cite{Molavi-20} examines a dynamic general equilibrium model in which an agent's choice in the current period affects the constraints he face in the future. By contrast, we focus on the dynamics of an agent's behavior in history-dependent learning problems instead of the steady states. We show that the long-run outcome can be inefficient by bounding the frequency of action switches.\footnote{\cite*{EPY2020} introduce stochastic approximation techniques and characterize the frequency of the agent's actions in misspecified learning problems without history-dependence.}

The attribution error in our model is related to 
\cite{eliaz2020model}, who study an agent's long-term behavior when he updates his belief according to a misspecified causal model. They propose a solution concept that characterizes the steady states of the above learning process, rather than examining the dynamics of actions and beliefs. 
\cite{S2013} examines the dynamic interaction between an agent and a sequence of principals, each of them acts only once and chooses whether to intervene. The agent attributes changes of a state variable to the latest intervention, which is applicable when some of the principal's actions (intervention) are more salient than others (no intervention). \cite{JS-12} characterize an informed long-run player's payoff and behavior when he faces a sequence of short-run players who mistakenly believe that all types of the long-run player use stationary strategies.
By contrast, we study a different type of attribution error, where the agent has wrong beliefs about the delay between actions and feedback.

\section{Model}\label{sec2}
Time is discrete, indexed by $t=1,2,..$. In period $t$, a Bayesian agent chooses an action $a_t \in A$, and then observes an outcome $y_t \in Y$. We assume that both $A$ and $Y$ are finite sets.

Our modeling innovation is to introduce time lags between decisions and feedback as well as the agent's misperception about the time lag. Formally, there exist two non-negative integers $k^*,k' \in \mathbb{N}$ with $k^* \neq k'$, such that the distribution of $y_t$ depends only on $a_{t-k^*}$,
while the agent believes that it depends on $a_{t-k'}$.

The agent faces uncertainty about the distribution over outcomes (which we call the \textit{state}) and learns about it over time by observing the history of actions and outcomes. A typical state is denoted by $F \equiv \{F(\cdot|a)\}_{a \in A}$, with $F(\cdot|a) \in \Delta (Y)$. Let $F^* \equiv \{F^*(\cdot|a)\}_{a \in A}$ be the \textit{true state}, namely, $y_t$ is distributed according to $F^*(\cdot|a_{t-k^*})$. The agent's prior belief about the state is $\pi_0 \in \Delta (\mathcal{Y})$, with $\mathcal{Y} \equiv \Big(\Delta (Y) \Big)^{A}$ and $\textrm{supp}(\pi_0)$ is finite.\footnote{When there are infinitely many states, \cite{diaconis1986consistency} and \cite{shalizi2009dynamics} show that the agent's posterior belief may not converge to the true state even when the true state belongs to the support of his prior belief. We abstract away from this complication in order to focus on the economic implications of misspecified beliefs about time lags.} After the agent learns that the state is $F$,  he believes that $y_t$ is distributed according to $F(\cdot|a_{t-k'})$.
The agent observes $h^t \equiv \{...,a_{-1},a_0,a_1,...,a_{t-1},y_1,....,y_{t-1}\}$ in period $t$ and his posterior belief is denoted by $\pi_t \in \Delta (\mathcal{Y})$. All the actions before period $1$ are exogenously given. 
In order to focus on misspecified belief about the time lag, we focus on prior beliefs that are \textit{regular}:
\begin{Condition}\label{cd:1}
$\pi_0$ is regular with respect to $F^*$ if
\begin{enumerate}
\item[1.]  $F^* \in \textrm{supp}(\pi_0)$, and for every $F \in \textrm{supp}(\pi_0)$ and $a \in A$, $F(\cdot|a)$ has full support.
\item[2.] for every $F,F' \in \textrm{supp}(\pi_0)$ and $a \in A$, we have $F(\cdot|a) \neq F'(\cdot|a)$.
\end{enumerate}
\end{Condition}
The first part  requires that the true state $F^*$ belongs to the support of the agent's prior belief and that the agent cannot rule out any state no matter which action he takes and which  outcome he observes. This rules out canonical forms of belief misspecifications studied by \cite{Ber-66}, \cite{Nya-19}, and \cite{EP-16} in which $F^*$ is excluded from the agent's prior belief.
The second part requires that the observed outcome is informative about the state regardless of the agent's action, which is satisfied for generic finite subsets of $\mathcal{Y}$. It rules out lack-of-identification problems, such as safe-arms in bandit models.

The agent's stage-game payoff is $v(y_t)$.
We assume that $\arg\max_{a \in A} \{ \sum_{y \in Y} v(y) F^*(y|a)\}$ is a singleton, and its unique element is denoted by $a^*$, i.e., the agent has a unique optimal action under the true state.
This is satisfied for generic $F^* \in \mathcal{Y}$ and $v: Y \rightarrow \mathbb{R}$ given that $A$ and $Y$ are finite sets. 
The agent's strategy is $\sigma: \mathcal{H} \rightarrow \Delta (A)$, where $\mathcal{H}$ is the set of histories. 
Strategy $\sigma$ is optimal if $\sigma(h^t)$ maximizes the expected value of $\sum_{s=0}^{+\infty} \delta^{s} v(y_{t+s})$ at every $h^t$, where $\delta \in [0,1)$ is the agent's discount factor.
We focus on settings such that either $\delta \in (0,1)$ or $(\delta,k')=(0,0)$. This is because the agent is indifferent between all actions when $\delta=0$ and $k' \geq 1$.
Let $\Sigma^*(\pi_0)$ be the set of  strategies that are optimal for the agent when his prior  is $\pi_0$.

For some useful benchmarks, the agent chooses $a^*$ in every period after he learns that the true state is $F^*$ even if he entertains a misspecified belief about the time lag. If there is no belief misspecification, i.e., $k^*=k'$, then according to Berk's Theorem \citep{Ber-66}, the agent's action converges to $a^*$ almost surely.

\paragraph{Remark:} Our baseline model focuses on situations in which the outcome in every period is affected only by one of the agent's actions. Section \ref{sec5} discusses extensions where the outcome in period $t$ depends on a convex combination of the agent's past and current-period actions, and the agent has a wrong belief about the weights of different actions. We also consider settings in which the agent faces uncertainty about the time lag and learns about it over time, but the support of his prior belief excludes the true time lag.

The agent's payoff in our baseline model depends only on the observed outcome.
Our results extend when the agent's payoff also depends on the state. 
When the agent's payoff depends directly on his action (e.g., different actions have different costs), his action can be suboptimal even when he learns the true state. 
This is because when $k^* \neq k'$, the agent
either overestimates or underestimates the time it takes for his action to have an effect, which can lead to suboptimal decisions 
since the agent discounts future payoffs. 
 Our results extend to settings where the agent's payoff is $v(y_t)-c(a_t)$, as long as the absolute value of $c(\cdot)$ is small enough such that the agent has a strict incentive to choose $a^*$ after he learns the true state. 

\section{Result}\label{sec3}
First, we show 
that if the agent's action converges, then it can only converge to his optimal action. Moreover, the asymptotic frequency of his optimal action must be strictly positive when his prior belief is regular. 
\begin{Lemma}\label{L2.1}
Suppose either $\delta \in (0,1)$ or $(\delta,k')=(0,0)$.
If $\pi_0$ is regular with respect to $F^*$ and $a_t$ converges to $a$ with positive probability, then $a=a^*$. Furthermore,
\begin{equation}\label{3.1}
    \liminf_{t \rightarrow +\infty}    \mathbb{E}^{\sigma} \Big[
\frac{1}{t} \sum_{s=1}^t \mathbf{1}\{a_s=a^*\} 
\Big] >0  \textrm{ for every } \sigma \in \Sigma^*(\pi_0).
\end{equation}
\end{Lemma}
The proof is in Appendix B. Intuitively, the only way in which misspecified beliefs about time lags can interfere learning is through an \textit{attribution error}, namely, the agent attributes the effects of $a_{t-k^*}$ to $a_{t-k'}$. 
 When the agent's action converges, $a_{t-k^*}$ and $a_{t-k'}$ are the same so the attribution error does not affect his learning. 
Since $F^* \in \textrm{supp}(\pi_0)$ and there is no lack-of identification problem, the agent will learn the true state almost surely. This implies that the agent's 
actions are asymptotically efficient, which contradicts the presumption that his action converges to something other than $a^*$.
The agent takes his optimal action with positive asymptotic frequency since
for every $\varepsilon>0$, the following event occurs with positive probability:
\begin{equation}\label{3.2}
   \mathcal{E}^{\varepsilon} \equiv \{\textrm{there exists } T \in \mathbb{N} \textrm{ such that } \pi_t(F^*) > 1-\varepsilon \textrm{ for every } t \geq T\}.
\end{equation}
Intuitively, this is because $y$ is informative about the state, so there always exists a signal realization that increases the posterior probability of $F^*$. 
The probability that the agent chooses $a^*$ in all future periods is strictly positive
when the posterior probability of $F^*$ is close to $1$,  which implies  (\ref{3.1}).

Despite the agent's action cannot converge to anything other than $a^*$, the attribution errors caused by  misspecified beliefs can lead to arbitrarily large long-term inefficiencies. Theorem \ref{Theorem1} shows that the frequency with which the agent takes his optimal action can be arbitrarily close to $0$. 
\begin{Theorem}\label{Theorem1}
Suppose either $\delta \in (0,1)$ or $(\delta,k')=(0,0)$.
For every $\gamma >0$, there exists an open set  $\mathcal{Y}^o \subset \mathcal{Y}$ such that for every $F^* \in \mathcal{Y}^o$,
there is
a prior belief $\pi_0$ that is regular with respect to $F^*$ under which 
\begin{equation}\label{eq:infinitesimal frequency} \limsup_{t \rightarrow +\infty}    \mathbb{E}^{\sigma} \Big[
\frac{1}{t} \sum_{s=1}^t \mathbf{1}\{a_s=a^*\}
\Big]   < \gamma \quad \textrm{for every} \quad \sigma \in \Sigma^*(\pi_0). 
\end{equation}
\end{Theorem}
Theorem \ref{Theorem1} implies that Bayesian agents fail to take their optimal action even when the true state belongs to the support of their prior belief and 
the observed outcome can statistically identify the state. 
Intuitively, this is because attribution errors lead to \textit{mislearning} when the agent switches actions. In particular, when the agent's action changes over time, the probability of the true state can decrease in expectation. When action switches are frequent enough, the amount of mislearning outweighs the what the agent learns when he takes the same action in consecutive periods. As a result, his posterior belief may attach a low probability to $F^*$ in all periods. Under some $F^*$ and $\pi_0$ that is regular with respect to $F^*$, such an event occurs with probability arbitrarily close to $1$, which leads to arbitrarily large asymptotic inefficiencies.

The proof is in Appendix C. We explain the logic behind our argument using an example,
which illustrates how attribution errors lead to action cycles in the long run and  how to bound the frequency of action switches using concentration inequalities.

\paragraph{Illustrative Example:} 
Suppose $A \equiv \{0,1\}$, $\delta=0$, $k^*=1$, and $k'=0$. That is, the agent is myopic, the distribution of $y_t$ depends only on the agent's action in period $t-1$ while the agent believes that $y_t$ is affected by his action in period $t$.
Let $Y \equiv \{y_0,y_1,y_2\}$, $v(y_0)=v(y_1)=0$, $v(y_2)=1$, and the support of $\pi_0$ is $\{F^*,F_0,F_1\}$, with
 \begin{align*}
  &F^*(y_i | a=0) = \begin{cases}
        \epsilon & i=0 \\
        1 - 3\epsilon & i=1 \\
        2\epsilon & i = 2
        \end{cases}\quad
        &F^*(y_i | a=1) = \begin{cases}
        1-2\epsilon & i=0 \\
        \epsilon & i=1 \\
        \epsilon & i = 2
        \end{cases}
    \end{align*}
    \begin{align*}
        &F_0(y_i | a=0) = \begin{cases}
        2/3 - \epsilon & i=0 \\
        1/3 - \epsilon & i=1 \\
        2\epsilon & i = 2
        \end{cases}\quad
        &F_0(y_i | a=1) = \begin{cases}
        2/3 - \epsilon/2 & i=0 \\
        1/3 - \epsilon/2 & i=1 \\
        \epsilon & i = 2
        \end{cases}
    \end{align*}
    \begin{align*}
        &F_1(y_i | a=0) = \begin{cases}
        1/3 - \epsilon/2 & i=0 \\
        2/3 - \epsilon/2 & i=1 \\
        \epsilon & i = 2
        \end{cases}\quad
        &F_1(y_i | a=1) = \begin{cases}
        1/3 - \epsilon & i=0 \\
        2/3 - \epsilon & i=1 \\
        2\epsilon & i = 2
        \end{cases}
        \end{align*}
These distributions are depicted in Figure 1. One can verify that the optimal actions in states $F^*$ and $F_0$ are both $0$, and the optimal action in state $F_1$ is $1$. 

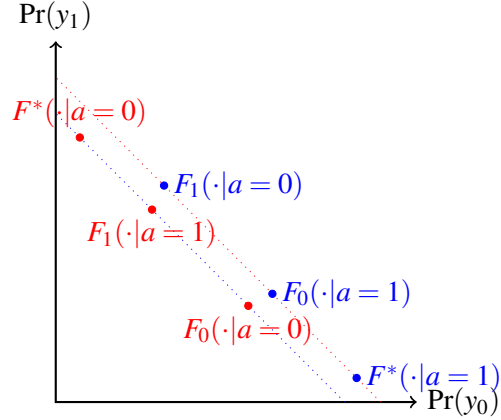
\begin{figure}
\begin{center}
\begin{tikzpicture}[scale=0.32]
\draw [<->, thick] (0,15)node[above]{$\Pr(y_1)$}--(0,0)--(15,0)node[right]{$\Pr(y_0)$};
\draw [dotted, red] (13.5,0)--(0,13.5);
\draw [dotted, blue] (12,0)--(0,12);
\draw [fill, red] (8,4) circle [radius=0.15];
\draw [fill, red] (4,8) circle [radius=0.15];
\draw [fill, red] (1,11) circle [radius=0.15];
\draw [fill, blue] (9,4.5) circle [radius=0.15];
\draw [fill, blue] (4.5,9) circle [radius=0.15];
\draw [fill, blue] (12.5,1) circle [radius=0.15];
\node [below, red] at (8,4) {$F_0(\cdot|a=0)$};
\node [below, red] at (4,8) {$F_1(\cdot|a=1)$};
\node [above, red] at (1,11) {$F^*(\cdot|a=0)$};
\node [right, blue] at (9,4.5) {$F_0(\cdot|a=1)$};
\node [right, blue] at (4.5,9) {$F_1(\cdot|a=0)$};
\node [right, blue] at (12.5,1) {$F^*(\cdot|a=1)$};
\end{tikzpicture}
\caption{Example with Action Cycles: Signal Distributions under $F^*$, $F_0$ and $F_1$}
\end{center}
\end{figure}
 
 Since the true state $F^*$ belongs to the support of $\pi_0$ and $y$ can statistically identify the state, the agent's action converges to $0$ almost surely when he has a correctly specified belief about the time lag. However, when the agent has a misspecified belief about the time lag, we sketch an argument which shows that the asymptotic frequency of the suboptimal action can be arbitrarily close to $1/2$.
 \begin{Claim}
 For every $\eta>0$,
there exists $\overline{\varepsilon} >0$ such that when $\varepsilon < \overline{\varepsilon}$, there exists a full support $\pi_0$ under which
\begin{equation*}
    \limsup_{t \rightarrow +\infty} \frac{1}{t} \mathbb{E}^{\sigma}\Big[
    \sum_{s=1}^t \mathbf{1}\{a_s=a^*\}
    \Big] < \frac{1}{2}+ \eta \quad \textrm{for every} \quad \sigma \in \Sigma^*(\pi_0). 
\end{equation*}
 \end{Claim}

Our argument proceeds in two steps. First, we examine an auxiliary learning problem where $F^*$ occurs with zero probability.
Let 
$l_t \equiv \log \frac{\pi_t(F_0)}{\pi_t(F_1)}$. By definition, the agent has a strict incentive to take action $0$ when $l_t > 0$, and has a strict incentive to take action $1$ when $l_t <0$. 

A useful observation from Figure 1 is that when the agent takes action $0$, $F_1$ is closer to $F^*$ compared to $F_0$. As a result, the log likelihood ratio $l_t$ decreases in expectation when the agent chooses action $0$ in two consecutive periods. Let $\tau_0$ be the number of periods with which the agent's action switches back to action~$1$ when $a_{t-2}=1$ and $a_{t-1}=a_t=0$. The Chernoff-Hoeffding inequality implies that:
\begin{equation}\label{3.4}
    \Pr(\tau_0 \geq s) \leq \Pr(l_{t+s} \geq 0) \leq \exp \Big(
    -2s \Big(
    \frac{l_t}{s} +\underbrace{\mathbb{E}[l_s-l_{s-1}]}_{<0}
    \Big)^2
    \Big),
\end{equation}
from which we know that the distribution of $\tau_0$ is first order stochastically dominated by an exponential distribution and therefore, has bounded first and second moments.

Similarly,
when the agent takes action $1$, $F_0$ is closer to $F^*$ compared to $F_1$. As a result, $l_t$ increases in expectation when the agent takes action $1$ in two consecutive periods.
Let  $\tau_1$ be the number of periods with which the agent's action switches back to action $0$ when $a_{t-2}=0$ and $a_{t-1}=a_t=1$. A similar argument based on the Chernoff-Hoeffding inequality implies that $\tau_1$ has bounded first and second moments.

In order to bound the frequency of action switches from below using the above conclusions on $\tau_0$ and $\tau_1$,  we establish a concentration inequality that applies to unbounded random variables  (Lemma \ref{lem:concentration}).\footnote{The Chernoff-Hoeffding inequality only applies to bounded random variables. Corollary 5.5 in \cite{lattimore2020bandit} and \cite{JLWZ-19} establish concentration inequalities for random variables with sub-Gaussian distributions. By contrast, our result is more general since it only requires the random variable to have bounded first and second moments.}
This inequality implies that
for every $\varepsilon>0$, there exists a large enough $T \in \mathbb{N}$ such that 
\begin{equation*}
\mathbb{E}\Big[\frac{\# \{t \leq T |a_{t-1}=a_t=1\}}{\# \{t\leq T| a_{t-1}=1, a_t=0\}}\Big]=\mathbb{E}\Big[\frac{\# \{t \leq T |a_{t-1}=a_t=1\}}{\# \{t\leq T| a_{t-1}=0, a_t=1\}} \Big] \leq \mathbb{E}[\tau_1]+\epsilon,
\end{equation*}
and
\begin{equation*}
\mathbb{E}\Big[ \frac{\# \{t \leq T |a_{t-1}=a_t=0\}}{\# \{t\leq T| a_{t-1}=1, a_t=0\}} \Big]=\mathbb{E}\Big[ \frac{\# \{t \leq T |a_{t-1}=a_t=0\}}{\# \{t\leq T| a_{t-1}=0, a_t=1\}} \Big] \leq \mathbb{E}[\tau_0] + \epsilon.
\end{equation*}
When $\pi_0(F_1)$ and $\pi_0(F_0)$ are close, the expectations of $\tau_0$ and $\tau_1$ are close, and therefore, the asymptotic frequencies of both actions are close to $1/2$. The above inequalities imply that the asymptotic frequency of action switches is strictly positive and is close to $\frac{1}{1+\mathbb{E}[\tau_0]}$ and $\frac{1}{1+\mathbb{E}[\tau_1]}$.

Next, we consider the case in which the true state $F^*$ belongs to the support of $\pi_0$ but occurs with low probability. 
We show that with probability close to $1$, the 
agent's posterior belief attaches a low probability to $F^*$ in all periods.
Formally, 
for every $\eta>0$, there exists $\overline{\pi}>0$ such that when the prior probability of $F^*$ is less than  $\overline{\pi}$, the probability of the event that
\begin{equation}\label{3.5}
 \max \Big\{   \frac{\pi_t(F^*)}{\pi_t(F_0)}, \frac{\pi_t(F^*)}{\pi_t(F_1)} \Big\} < \eta \textrm{ for all } t \in \mathbb{N}
\end{equation}
is at least $1-\eta$. Intuitively, both $\log \frac{\pi_t(F^*)}{\pi_t(F_0)}$ and $\log \frac{\pi_t(F^*)}{\pi_t(F_1)}$
increase in expectation when $a_{t}=a_{t-1}$ since $F^*$ is the true state. 
However, as can be seen from Figure 1, $F^*(\cdot|a=0)$ is further away from $F^*(\cdot|a=1)$ compared to both $F_0(\cdot|a=0)$ and $F_1(\cdot|a=0)$, and 
$F^*(\cdot|a=1)$ is further away from $F^*(\cdot|a=0)$ compared to both $F_0(\cdot|a=1)$ and $F_1(\cdot|a=1)$. 
Due to the attribution errors caused by misspecified beliefs about the time lag, both log likelihood ratios \textit{decrease} in expectation when  $a_{t} \neq a_{t-1}$.

When $\pi_t(F^*)$ is low, the agent's best reply problem is similar to the one he faces in the auxiliary scenario where $F^*$ is excluded from the support of his prior, in which case he frequently switches actions. Those action switches together with the attribution error lead to mislearning. 
When $\mathbb{E}[\tau_0]$ and $\mathbb{E}[\tau_1]$ are sufficiently small, action switches are frequent, so the mislearning caused by attribution errors outweighs what the agent can learn when he takes the same action in adjacent periods.

In order to formalize this intuition, we provide a lower bound on the probability of event (\ref{3.5}) using concentration inequalities. 
The Chernoff-Hoeffding inequality does not apply
since the agent's belief affects his actions, so the log likelihood ratios between $F^*$ and $F_0$, and between $F^*$ and $F_1$  can exhibit serial correlations.  
We overcome this challenge by constructing  a martingale process with bounded increments from the log likelihood ratios and then applying the Azuma-Hoeffding inequality. We show that with probability close to $1$, the agent's belief attaches low probability to $F^*$ in all periods. This together with the frequent action switches explains why he takes the inefficient action with positive asymptotic frequency and his actions cycle over time.

\paragraph{Remark:} In our example, the asymptotic frequency with which the agent takes the inefficient action is close to $1/2$. In Appendix \ref{secB}, we construct $F^*$ as well as regular prior beliefs with respect to $F^*$ such that the frequency of the inefficient action is close to $1$. In this example, one can simply modify the outcome distributions such that both $\tau_1$ and $\tau_0$ have low expectations, but the expectation of $\tau_1$ is significantly greater than the expectation of $\tau_0$.

\section{Application: Dynamic Policy Choice Game}\label{sec4}
In order to demonstrate the applicability of our techniques to bound the frequency of action switches,
we analyze a dynamic policy choice game between
\begin{itemize}
    \item a principal who strategically makes policy proposals, learns about the state over time, and has a correctly specified belief about the time lag between the chosen policy and the observed feedback,
    \item a  Bayesian agent who can veto the principal's proposals, and learn about the outcome distribution under a misspecified belief about the time lag.\footnote{Our analysis also applies to a sequence of myopic agents, each plays the game only once.}
\end{itemize}

In every period, a society needs to make a collective choice between two policies $a_t \in A \equiv \{0,1\}$. In period $t$, the principal makes a proposal $\widetilde{a}_t \in A$. If $\widetilde{a}_t=0$, then action $0$ is automatically implemented, i.e., $a_t=0$. If $\widetilde{a}_t=1$, then the agent chooses whether to accept ($a_t=1$) or veto ($a_t=0$) the principal's proposal.

Both the principal and the agent face uncertainty about the \textit{state}, which is contained in  $\mathcal{F} \equiv \{F_0,F_1\}$. 
Their common prior belief is $\pi_0 \in \Delta (\mathcal{F})$, which we assume has full support.\footnote{Extensions to environments with more than two states are available upon request. Our results also apply when the principal and the agent agree to disagree about the state distribution. For example, when the principal's prior belief is $p_0$ that is different from $\pi_0$, one needs to replace $\pi_0$ by $p_0$ in RHS of (\ref{4.17}).}

The principal and the agent agree to disagree in terms of the time lag between decisions and feedback. 
The principal has a correctly specified belief about the time lag and knows that the distribution of $y_t$ depends only on $a_{t-k^*}$, with $k^* \geq 1$. The agent believes that the distribution of $y_t$ depends on $a_t$, i.e., $k'=0$. In period $t$, 
both players observe $h^t \equiv \{\widetilde{a}_s,a_s,y_s\}_{s=0}^{t-1}$ and update their beliefs about the state according to Bayes rule. Let $\mathcal{H}$ be the set of histories. Let $\sigma_p: \mathcal{H} \rightarrow [0,1]$ be the principal's strategy, which maps the histories to the probability that he proposes action $1$, with $\sigma_p \in \Sigma_p$.  Let $\sigma_a: \mathcal{H} \rightarrow [0,1]$ be the agent's strategy, which maps the histories to the probability with which he approves action $1$, with $\sigma_a \in \Sigma_a$.

The principal is patient and maximizes the frequency of action $1$.\footnote{We evaluate the patient principal's payoff using the long-run averages. This is a common practice in undiscounted games, see for example, \citet{Har-85} and \citet{For-92}.}
The agent is myopic and his stage-game payoff in period $t$ is $v(y_t)$.\footnote{We comment on the case in which $k' \neq 0$ and the agent's discount factor is strictly positive by the end of this section.} Since the principal has no private information about the state,
under a \textit{no signaling what you don't know} condition \citep{FT-91},
neither 
the agent's belief nor his best reply depends on the principal's proposals $\{\widetilde{a}_0,...,\widetilde{a}_t\}$ or the principal's strategy $\sigma_p$. Without loss of generality, we assume that  $\sum_{y \in Y} v(y) F_1(y|1) >  \sum_{y \in Y} v(y) F_1(y|0)$ and $\sum_{y \in Y} v(y) F_0(y|0) >  \sum_{y \in Y} v(y) F_0(y|1)$, that is, action $a$ is optimal for the agent in state $F_a$ for every $a \in A$.\footnote{If action $1$ is optimal for the agent in both states, then the principal can implement action $1$ with frequency $1$ regardless of the agent's prior belief and belief misspecification. If action $0$ is optimal for the agent in both states, then the frequency of action $1$ is zero regardless of the principal's strategy.}

 This game fits applications where a benevolent policy-maker (i.e., the principal) wants to persuade the public (i.e., the agents) to stop taking actions that have negative externalities on marginalized groups (action $0$, for example, massive gatherings during a pandemic have negative externalities on people who are immunocompromised), or to adopt reforms that have positive externalities (action $1$, for example, reducing greenhouse gas emission has positive externalities on future generations).
Action $0$ is interpreted as a \textit{status quo action}, which the policy-maker has the ability to implement by himself. By contrast, the public's cooperation is crucial for the implementation of the socially beneficial action. For example, the government can issue a mask mandate for the purpose of slowing down the spread of a virus, but this mandate won't be effective unless the majority of citizens cooperate. However, taking the socially beneficial action is against the agent's private interest in state $F_0$ and he learns about which action is optimal over time.

Our result characterizes the maximal frequency that the principal can implement the socially beneficial action by taking advantage of the agent's misspecified beliefs. We also describe the qualitative features of the principal's strategy from which he approximately attains his optimal payoff. We assume that the agent is not indifferent between action $0$ and action $1$ at any history. 
\begin{Assumption}\label{Ass1}
$\pi_0$ is such that the agent is not indifferent between action $0$ and action $1$ at every $h^t$.
\end{Assumption}
Assumption \ref{Ass1} is satisfied for generic prior belief $\pi_0$ given that $A$ and $Y$ are finite sets. This assumption implies that the agent's optimal strategy is unique, which we denote by $\sigma_a^*$. The principal's asymptotic payoff from strategy $\sigma_p$ is between
\begin{equation}\label{4.1}
\underline{V}(\sigma_p) \equiv   \liminf_{t \rightarrow +\infty} \frac{1}{t}\mathbb{E}^{(\sigma_p,\sigma_a^*)}\Big[\sum_{s=1}^{t} a_s \Big]
\quad \textrm{and} \quad
\overline{V}(\sigma_p) \equiv   \limsup_{t \rightarrow +\infty} \frac{1}{t}\mathbb{E}^{(\sigma_p,\sigma_a^*)}\Big[\sum_{s=1}^{t} a_s \Big],
\end{equation}
where $\mathbb{E}^{(\sigma_p,\sigma_a^*)}[\cdot]$ is the expectation under $(\sigma_p,\sigma_a^*)$.
The principal's payoff when he optimally chooses his strategy is bounded between
$\underline{V} \equiv   \sup_{\sigma_p \in \Sigma_p} \underline{V}(\sigma_p)$
and  $\overline{V} \equiv \sup_{\sigma_p \in \Sigma_p} \overline{V}(\sigma_p)$.

We introduce some notation to characterize the principal's optimal payoff.
Let
\begin{equation*}
    \mathcal{E}^* \equiv \Big\{
    \textrm{ there exists } t \in \mathbb{N}
    \textrm{ such that for every } s \geq t,
  \sum_{F \in \mathcal{F}} \pi_s(F) \sum_{y \in Y} v(y)F(y|1)
  >
   \sum_{F \in \mathcal{F}} \pi_s(F) \sum_{y \in Y} v(y)F(y|0)
    \Big\}
\end{equation*}
be the event that action $1$ is strictly optimal for the agent starting from some period. Let
\begin{equation}\label{4.13}
    q^* \equiv \sup_{\sigma_p \in \Sigma_p} \Pr(\mathcal{E}^*|\sigma_p,\sigma_a^*,F_0).
\end{equation}
Intuitively, $q^*$ is the maximal probability of event $\mathcal{E}^*$ when the state is $F_0$ and the agent plays according to his optimal strategy.
Let $X_{a \rightarrow a'}$ be a random variable such that
\begin{equation}\label{4.14}
    X_{a \rightarrow a'}=\log \frac{F_1(y|a')}{F_0(y|a')} \textrm{ with probability } F_0(y|a) \textrm{ for every } y \in Y.
\end{equation}
Intuitively, $X_{a \rightarrow a'}$ is the change in the log likelihood ratio between $F_0$ and $F_1$ when the true state is $F_0$, the previous period action was $a$, and the current period action is $a'$. 
Let
\begin{align}\label{lambda}
\lambda &\equiv 
\sup_{\widehat{\lambda}\geq 0} 
\frac{\widehat{\lambda}+1}{\hat{\lambda}+2} 
\end{align}
subject to
\begin{align}\label{constraint}
\widehat{\lambda} \mathbb{E}[X_{1 \rightarrow 1} ] + \mathbb{E}[X_{1 \rightarrow 0}  +X_{0 \rightarrow 1} ] > 0.
\end{align}
Moreover, we define $\lambda = 0$ if the above inequality is never satisfied. 
One can verify that $\lambda \in (\frac{1}{2},1)\cup\{0\}$. 
Intuitively, $\lambda$ is the maximal frequency of action $1$ such that the log likelihood ratio between $F_0$ and $F_1$ does not increase in expectation, or in another word, the amount of mislearning in state $F_0$ is non-negative. \begin{Theorem}\label{Theorem2}
If $\pi_0$ satisfies Assumption \ref{Ass1}, then
\begin{equation}\label{4.17}
    \overline{V}=\underline{V}= \pi_0(F_1) + \pi_0(F_0)  q^* \lambda.
\end{equation}
\end{Theorem}
Theorem \ref{Theorem2} implies that the principal's asymptotic payoff exists (i.e., $\overline{V}=\underline{V}$)
and characterizes its value. At the optimum, the asymptotic frequency of action $1$ is $1$ in state $F_1$ and is $q^*\lambda$ in state $F_0$.

The proof is in Appendix \ref{secD}. 
We provide an intuitive explanation in three steps, using an example in which there are two outcomes $Y=\{y_g,y_b\}$, the outcome distributions are given by:
\begin{equation*}
F_1(y|a) \equiv \left\{ \begin{array}{ll}
r & \textrm{ if } (a,y)=(0,y_b) \textrm{ or } (1,y_g) \\
1-r & \textrm{ if }  (a,y)=(1,y_b) \textrm{ or } (0,y_g),
\end{array} \right.
\end{equation*}
\begin{equation*}
F_0(y|a) \equiv \left\{ \begin{array}{ll}
r & \textrm{ if } (a,y)=(1,y_b) \textrm{ or } (0,y_g) \\
1-r & \textrm{ if }  (a,y)=(0,y_b) \textrm{ or } (1,y_g),
\end{array} \right.
\end{equation*}
where $r \in (1/2,1)$ is a parameter, and the agent's payoff is $1$ when the outcome is $y_g$ and is $0$ otherwise.

In this example, the agent strictly prefers action $0$ if and only if $\log \frac{\pi_t(F_0)}{\pi_t (F_1)}>0$.
One can verify that $X_{1 \rightarrow 1}$ first order stochastically dominates both 
$X_{1 \rightarrow 0}$ and
$X_{0 \rightarrow 1}$. Therefore, the maximum that defines $q^*$ is attained
when the principal proposes
the opposite action to what was implemented $k^*$ periods ago, that is, $\widetilde{a}_t= 1-a_{t-k^*}$ for every $t \in \mathbb{N}$.
According to the maximization problem that defines $\lambda$, 
\begin{equation}\label{4.6}
\lambda=\max_{\widehat{\lambda}} \frac{\widehat{\lambda}+1}{\widehat{\lambda}+2} \quad \textrm{subject to} \quad
\widehat{\lambda} \mathbb{E}[X_{1 \rightarrow 1}]+ \mathbb{E}[X_{1 \rightarrow 0}+X_{0 \rightarrow 1}] \geq 0.
\end{equation}
Since
 $X_{1 \rightarrow 1}$ first order stochastically dominates both 
$X_{1 \rightarrow 0}$ and
$X_{0 \rightarrow 1}$, 
the constraint is binding and
the maximum in (\ref{4.6}) is attained when the ratio between taking the same action in consecutive periods and action switches is $\widehat{\lambda}$.

\paragraph{Step 1:} 
We consider an \textit{auxiliary game} in which the principal knows the true state but the agent is naive in the sense that he fails to extract information from the principal's proposals.
We show that when the true state belongs to $\mathcal{F}_1$,  the principal's asymptotic payoff in the auxiliary game is $1$ regardless of the agent's prior belief. Moreover, the principal can attain this payoff by proposing action $1$ in every period.


Let $l_t \equiv \log \frac{\pi_t(F_0)}{\pi_t (F_1)}$. Agent $t$ strictly prefers action $1$ when $l_t<0$ and strictly prefers action $0$ when $l_t >0$.
Since $l_t$ decreases in expectation when the agent takes the same action in two consecutive periods and $F_1$ is the true state, the Wald's inequality (Lemma \ref{lem:wald}) implies that
\begin{enumerate}
    \item for every $l_t <0$, the probability of the event  $\{l_{\tau} < 0 \textrm{ for every } \tau \geq t\}$ is bounded away from $0$,
    \item for every $l_t >0$, the probability of the event  $\{l_{\tau} > 0 \textrm{ for every } \tau \geq t\}$ is $0$.
\end{enumerate}
Since the log likelihood ratio process is absorbed  with positive probability at negative values and is absorbed  with zero probability  at positive values, 
we know that with probability $1$, there exists $T \in \mathbb{N}$ such that $l_t <0$ for all $t \geq T$. Therefore, the principal's asymptotic payoff is $1$ regardless of the agent's prior belief.

\paragraph{Step 2:} We show that principal's optimal payoff in the auxiliary game where the state is $F_0$ is $q^*\lambda$, that is,
\begin{equation}\label{4.9}
U(F_0,\pi_0)=q^* \lambda.
\end{equation}

Recall that (1) $q^*$ is the maximal probability that the agent eventually has an incentive to approve action $1$, which is attained when the principal proposes action $1$ with frequency approximately $1/2$ the agent knows event $\mathcal{E}^*$, and (2) according to (\ref{4.6}), $\lambda$ is the maximal frequency that the principal can propose action $1$ subject to the constraint that the log likelihood ratio between $F_0$ and $F_1$ does not increase in expectation.

The principal faces a tradeoff between increasing the frequency that he proposes action $1$ and increasing the probability that the agent is willing to approve action $1$. The former allows him to propose action $1$ with frequency as high as $\lambda$, but in order to maximize the probability with which the agent approves action $1$, he needs to propose it with frequency close to $1/2$.

Equation (\ref{4.9}) suggests that such a tradeoff has no impact on the principal's asymptotic payoff since in the auxiliary game where the state is $F_0$, the principal can attain an expected payoff as if (1) the agent eventually approves action $1$ for all periods with its maximal probability $q^*$, and (2) the principal can propose action $1$ with its maximal frequency subject to the mislearning constraint, which equals $\lambda$.

The definitions of $q^*$ and $\lambda$ imply that $U(F_0,\pi_0) \leq q^* \lambda$.
We show that $U(F_0,\pi_0) \geq q^* \lambda$ by constructing a family of strategies under which the principal's asymptotic payoff is arbitrarily close to $q^*\lambda$. Each strategy in this class is characterized by a cutoff log likelihood ratio $l_{\varepsilon}^*<0$ such that
the principal proposes $\widetilde{a}_t=1-a_{t-k^*}$ if the log likelihood ratio is above $l_{\varepsilon}^*$ in all previous period, and proposes action $1$ with frequency close to but less than $\lambda$ when the log likelihood ratio has fall below $l_{\varepsilon}^*$ in at least one period.

The key step is to show that under the proposed strategy (1) the probability with which the log likelihood ratio falls below $l_{\varepsilon}^*$ in at least one period is arbitrarily close to $q^*$, and (2) conditional on the log likelihood ratio falls below $l_{\varepsilon}^*$, the probability that it is strictly negative in all future periods is close to $1$.

We establish these two claims using concentration inequalities.
The intuition behind the first claim is that conditional on $l_t>0$, the probability that the log likelihood ratio is  positive in all future periods is strictly positive, so the log likelihood ratio will eventually escape any bounded interval with probability~$1$. As a result, the probability that $l_t<0$ for all $t$ large enough equals the probability that $l_t<l^*_{\varepsilon}$ for all $t$ large enough.
The intuition behind the second statement is that according to (\ref{4.9}), one can construct strategies under which the frequency of proposing action $1$ is close to $\lambda$, yet the log likelihood ratio is non-increasing in expectation. For an example of such a strategy, let $T_1, T_2 \in \mathbb{N}$ be such that $T_1$ is even and $T_2/T_1 \in (\lambda-\varepsilon,\lambda)$. The principal's strategy is divided into $T \equiv T_1+T_2$ period blocks such that he proposes action $0$ in period $1$, $3$, ... $T_1-1$ within each block, and proposes action $1$ otherwise.
The Wald's inequality implies that the probability with which the log likelihood ratio exceeds $0$ in some period after $t$ is small when $l_t < l_{\varepsilon}^*$ and $l_{\varepsilon}^*$ is small enough.

To conclude, when the principal uses this class of strategies, he can ensure that with probability close to $q^*$, the agent is willing to approve action $1$ in all future periods, and conditional on this event, he can propose policy $1$ with frequency arbitrarily close to $\lambda$. This explains why the tradeoff he faces between inducing mislearning and increasing the frequency of proposing action $1$ diminishes in the long run.

\paragraph{Step 3:} We show that the principal's payoff in our dynamic policy choice game with symmetric uncertainty equals his expected payoff in the auxiliary game. Formally, let $U(F,\pi_0)$ be the principal's payoff in the auxiliary game when the state is $F$ and the agent's prior belief is $\pi_0$. We show that:
\begin{equation}\label{4.8}
   \underline{V}=\overline{V}=  \pi_0(F_0)U(F_0,\pi_0)+\pi_0(F_1) U(F_1,\pi_0).
\end{equation}
This is implied by the following two inequalities:
\begin{equation}\label{4.10}
  \overline{V} \leq \pi_0(F_0)U(F_0,\pi_0)+\pi_0(F_1) U(F_1,\pi_0)
\end{equation}
and
\begin{equation}\label{4.11}
   \underline{V} \geq \pi_0(F_0)U(F_0,\pi_0)+\pi_0(F_1) U(F_1,\pi_0).
\end{equation}
Inequality (\ref{4.10}) is straightforward since the principal's payoff is weakly greater in the auxiliary game given that he has more information and the agent does not extract information from his proposals.

In order to establish inequality (\ref{4.11}), let $\sigma_{p}^{\varepsilon}$ be the principal's strategy such that his asymptotic payoff is more than $U(F_0,\pi_0)-\varepsilon$ in the auxiliary game where the state is $F_0$. Since $y$ is informative about the state, for every $\varepsilon>0$, there exists $T \in \mathbb{N}$ such that for each of the principal's strategy and every $a \in A$, the principal's posterior belief in period $T$ attaches probability greater than $1-\varepsilon$ to $F_a$ when $F_a$ is the true state.

Consider the principal's asymptotic payoff by using the following strategy in the original game with symmetric uncertainty:
\begin{enumerate}
    \item he plays according to $\sigma_p^{\varepsilon}$ in the first $T$ periods,
    \item if his period $T$ posterior belief attaches probability greater than $1-\varepsilon$ to state $F_1$, then he proposes action $1$ in all future periods,
    \item if his period $T$ posterior belief attaches probability less than $1-\varepsilon$ to state $F_1$, he continues to use strategy $\sigma_p^{\varepsilon}$.
\end{enumerate}
Under the above strategy, the probability with which the principal proposes $1$ in every period after $T$ is at least $1-\varepsilon$ conditional on the state being $F_1$. 
Since $U(F_1,\pi_0)=1$ for all $\pi_0$,
the principal's asymptotic payoff conditional on state $F_1$ is at least $1-\varepsilon$. Conditional on the true state being $F_0$, the probability with which he uses $\sigma_p^{\varepsilon}$ in every period is at least $1-\varepsilon$, so his payoff is at least $(1-\varepsilon)(U(F_0,\pi_0) -\varepsilon )$. Therefore, his expected asymptotic payoff converges to $\pi_0(F_0) U(F_0,\pi_0) +\pi_0(F_1)$ as $\varepsilon$ goes to $0$.

\paragraph{Remarks:} Our formula for the principal's optimal payoff is reminiscent of the repeated zero sum games of \cite{aumann1995repeated} and the Bayesian persuasion games of \cite{kamenica2011bayesian}, where an informed player's payoff in a binary-state setting is a piece-wise linear and concave function of the uninformed player's prior belief. Our formula for the principal's highest equilibrium payoff (\ref{4.17}) is not continuous since $q^*$ depends on the agent's prior belief and exhibits discontinuity in general.

When $\delta>0$, the agent may have incentives to experiment, which depend on the principal's  strategy. As a result, the agent's optimal strategy is not unique when the log likelihood ratio between $F_0$ and $F_1$ is close to the cutoff at which a myopic agent is indifferent. In general, for every $\delta \in (0,1)$ and $k'$, there exist two cutoffs $l^*$ and $l^{**}$ with $0<l^*<l^{**}<+\infty$ such that regardless of the principal's strategy, the agent has a strict incentive to approve action $1$ when $l_t < l^*$ and has a strict incentive to veto action $1$ when $l_t > l^{**}$. When $l_t \in [l^*,l^{**}]$, the agent's incentive depends on the principal's strategy. Nevertheless, when the discount factor is positive but small enough, the set of agent-optimal strategy is small and our approach provides lower and upper bounds on the principal's payoff.  The two bounds coincide as the agent's discount factor converges to $0$, in which case our approach can exactly characterize the principal's asymptotic payoff.

\section{Discussions}\label{sec5}
We discuss extensions and generalizations of our main result.

\paragraph{Uncertainty about the time lag:} In our baseline model, the agent faces uncertainty about the outcome distribution but has a degenerate prior about the time lag. In general, the agent may also learn about the time lag under a misspecified model.  

Our main result extends when the agent faces uncertainty both about the outcome distribution and the time lag.\footnote{If the agent only faces uncertainty about the time lag but knows the outcome distribution, then he takes his optimal action in every period and misspecified belief about the time lag is irrelevant for his behavior and payoff.} Formally, there is a finite set of states $\mathcal{F} \subset \mathcal{Y} $ and a finite set of possible time lags $K \subset \mathbb{N}$. The agent has a full support prior belief $\pi_0 \in \Delta (\mathcal{F} \times K)$.
In order to focus on the effects of misspecified belief about the time lag, we assume that $F^* \in \mathcal{F}$ and $k^* \notin K$. 

Similar to the baseline model, the agent chooses $a^* \equiv \arg\max_{a \in A} \sum_{y \in Y} v(y)F^*(y|a)$ in every period when he learns the true state regardless of his belief about the time lag. As a result, the agent's action cannot converge to actions other than $a^*$ and $a^*$ occurs with positive asymptotic frequency. When $\delta \in (0,1)$ or $(\delta,k')=(0,0)$, there exists $F^* \in \mathcal{Y}$ and a prior belief $\pi_0$ that is regular with respect to $F^*$ under which the agent takes his optimal action with frequency arbitrarily close to $0$.

\paragraph{General forms of belief misspecifications:} In our baseline model, the distribution of $y_t$ depends only on one of the agent's actions. In practice, the outcome distribution can be affected by multiple actions.

We extend our results when the distribution of $y_t$ depends on a convex combination of the agent's current-period action and his actions in the last $k \in \mathbb{N}$ periods where $k$ is an exogenous parameter. In particular, when the state is $F \equiv \{F(\cdot|a)\}_{a \in A}$, $y_t$ is distributed according to $\sum_{a \in A} \alpha_t(a) F(\cdot|a)$ where
$\alpha_t(a) \equiv \sum_{j=0}^k \beta_j \mathbf{1}\{a_{t-j}=a\}$, $\boldsymbol{\beta} \equiv (\beta_0,...,\beta_k) \in \mathbb{R}_+^{k+1}$, and $\sum_{j=0}^k \beta_j =1$. The agent has a wrong belief about the convex weights of different actions, and believes that when the state is $F$, $y_t$ is distributed according to $\sum_{a \in A} \widehat{\alpha}_t(a) F(\cdot|a)$, where
$\widehat{\alpha}_t (a) \equiv \sum_{j=0}^k \widehat{\beta}_j \mathbf{1}\{ a_{t-j}=a\}$ with $\boldsymbol{\widehat{\beta}} \equiv (\widehat{\beta}_0,...,\widehat{\beta}_k) \in \mathbb{R}_+^{k+1}$, $\sum_{j=0}^k \widehat{\beta}_j =1$, and $\boldsymbol{\widehat{\beta}}  \neq \boldsymbol{\beta}$. This general formulation captures for example, when the agent overestimates or underestimates the effects of his current-period action on the current-period outcome, i.e., when $\widehat{\beta}_0 \neq \beta_0$.
Theorem \ref{Theorem1} extends under a stronger identification condition that for every $F,F' \in \mathcal{F}$ and $\alpha \in \Delta (A)$, we have $F(\cdot|\alpha) \neq F'(\cdot|\alpha)$.

\newpage
\appendix
\input{proof3}

\end{spacing}
\bibliographystyle{plainnat}
\bibliography{ref}
\end{document}

%% file: math.tex
\usepackage{ifthen}

\DeclareMathOperator{\argmax}{argmax}

\newcommand{\given}{\,\mid\,}

\newcommand{\prob}[2][]{\text{\bf Pr}\ifthenelse{\not\equal{}{#1}}{_{#1}}{}\!\left[{\def\givenn{\middle|}#2}\right]}
\newcommand{\expect}[2][]{{\mathbb{E}}\ifthenelse{\not\equal{}{#1}}{_{#1}}{}\!\left[{\def\givenn{\middle|}#2}\right]}

\newcommand{\tparen}{\big}
\newcommand{\tprob}[2][]{\text{\bf Pr}\ifthenelse{\not\equal{}{#1}}{_{#1}}{}\tparen[{\def\given{\tparen|}#2}\tparen]}
\newcommand{\texpect}[2][]{\text{\bf E}\ifthenelse{\not\equal{}{#1}}{_{#1}}{}\tparen[{\def\given{\tparen|}#2}\tparen]}

\newcommand{\sprob}[2][]{\text{\bf Pr}\ifthenelse{\not\equal{}{#1}}{_{#1}}{}[#2]}
\newcommand{\sexpect}[2][]{\text{\bf E}\ifthenelse{\not\equal{}{#1}}{_{#1}}{}[#2]}

\newcommand{\strategy}{\sigma}

\newcommand{\event}{{\cal E}}

\newcommand{\cF}{{\cal F}}
\newcommand{\az}{0}
\newcommand{\ao}{1}

\newcommand{\bE}{\mathbb{E}}

\newcommand{\qstar}{q^*}
\newcommand{\lambdaf}{\lambda}

%% file: proof3.tex
\section{Probability Tools}\label{secA}
We state three results in probability theory, which will be used in our subsequent proofs. The first result is the Wald nequality, which bounds the probability of the union of tail events from above.
\begin{Lemma}[\citealp{Wald-44}]\label{lem:wald}
Let $\{Z_t\}_{t \in \mathbb{N}}$ be a sequence of i.i.d. random variables with finite support, strictly negative mean, and takes a positive value with positive probability. Let
 $r^* > 0$ be the unique real number that satisfies 
$\expect[z\sim Z_1]{\exp(r^*z)} = 1$.
We have
\begin{align*}
\Pr\left[\bigcup_{n=1}^\infty \left\{\sum_{t=1}^{n} Z_t \geq c\right\}\right]
\leq \exp(-r^* \cdot c) \textrm{ for every } c>0.
\end{align*}
\end{Lemma}
The second result is the Azuma-Hoeffding inequality, that applies to martingales with bounded increments. 
\begin{Lemma}[Azuma-Hoeffding inequality]\label{lem:azuma}
Let $\left\{Z_{0},Z_{1},\cdots \right\}$ be a martingale such that $|Z_{k}-Z_{k-1}|\leq c_{k}$.
Then for every $N \in \mathbb{N}$ and $\epsilon_1 > 0$, we have
$$\Pr[Z_{N}-Z_{0}\geq \epsilon_1 ]\leq \exp \left(-\frac{\epsilon_1 ^{2}}{2\sum _{k=1}^{N}c_{k}^{2}}\right).$$
\end{Lemma}
The third result extends the Chernoff-Heoffding inequality to random variables with unbounded support and finite first and second moments.

\begin{Lemma}\label{lem:concentration}
For any $\lambda > 1$ and any sequence of i.i.d. random variables $X_1, X_2, \dots, X_n$ with finite mean $\mu>0$ and finite variance,
we have 
\begin{align*}
\Pr\left[\sum_{i=1}^n X_i \geq \lambda \mu n\right] \leq \exp(-cn)
\end{align*}
where $c \equiv \max_{t>0}\{\lambda\mu t - \ln\bE\left[e^{t X_1}\right]\}$. 
For any $\lambda \in (0, 1)$, we have 
\begin{align*}
\Pr\left[\sum_{i=1}^n X_i \leq \lambda \mu n\right] \leq \exp(-c'n)
\end{align*}
where $c' \equiv \max_{t>0}\{-\lambda\mu t - \ln\bE\left[e^{-t X_1}\right]\}$
\end{Lemma}
\begin{proof}
The proof is similar to that of the Chernoff-Hoeffding inequality. If $\lambda > 1$, then
\[
\Pr\left[\sum_{i=1}^n X_i  \geq \lambda n \mu\right]
= \Pr\left[e^{t\sum_{i=1}^n X_i } \geq e^{t \lambda n \mu}\right]
\leq \frac{1}{e^{t \lambda n \mu}} \cdot \prod\limits_{i = 1}^n \bE\left[e^{t X_i}\right] \textrm{ for every } t \in \mathbb{N},
\]
where the last inequality holds by the Markov inequality.
Set $t \in \mathbb{N}$ in order to maximize $\lambda\mu t - \ln\bE\left[e^{t X_1}\right]$, we have 
\begin{equation*}
\Pr\left[\sum_{i=1}^n X_i  \geq \lambda n \mu\right]
\leq \exp\left\{-cn\right\}.
\end{equation*}
If $\lambda < 1$, then
\[
\Pr\left[\sum_{i=1}^n X_i \leq \lambda n \mu\right]
= \Pr\left[e^{-t\sum_{i=1}^n X_i } \geq e^{-t \lambda n \mu}\right]
\leq \frac{1}{e^{-t \lambda n \mu}} \cdot \prod\limits_{i = 1}^n \bE\left[e^{-t X_i}\right] \textrm{ for every } t \in \mathbb{N},
\]
Set $t \in \mathbb{N}$ in order to maximize $- \lambda\mu t - \ln\bE\left[e^{-t X_1}\right]$, we have 
\begin{equation*}
\Pr\left[\sum_{i=1}^n X_i  \leq \lambda n \mu\right]
\leq \exp\left\{-c'n\right\}.
\end{equation*}
\end{proof}

\paragraph{Remark: } In Lemma \ref{lem:concentration}, 
let $c(t) = \lambda\mu t - \ln\bE\left[e^{t X_1}\right]$. 
One can verify that $c(0) = 0$ and $c'(0) = (\lambda-1)\mu > 0$ for $\lambda > 1$ and $\mu>0$. 
Moreover, $c''(0) = {\rm Var}[X_1]$ is finite. 
Therefore, $c = \max_{t>0}\{\lambda\mu t - \ln\bE\left[e^{t X_1}\right]\}$ is strictly positive for $\lambda > 1$ and $\mu>0$. 
Similarly, we can also show that $c'>0$ for $\lambda < 1$ and $\mu>0$.

\section{Proof of Lemma 3.1}\label{secB}
The conclusion of Lemma \ref{L2.1} is implied by the following two claims. 
\begin{Claim}\label{clm:b1}
For $\epsilon \in (0,1)$ and $\pi_0$ that has finite support, there exists $\eta > 0$ and a sufficiently large $T$ such that the $\pi_{T}(F^*) > 1-\epsilon$ with probability at least $\eta$.
\end{Claim}

\begin{Claim}\label{clm:b2}
Suppose $\delta \in (0,1)$ or $(\delta,k')=(0,0)$.
For any finite set of states $\cF\subseteq \Delta (\mathcal{Y})$, 
there exist $\epsilon, \eta \in (0,1)$
such that if $\pi_0(F^*) > 1-\epsilon$ and $\pi_0$ is supported on $\cF$, 
then $a^*$ is chosen for all periods with probability at least $\eta$. 
\end{Claim}

Combining those two claims, we have that 
for any prior $\pi$ with finite support, 
there exists $\eta > 0$ and $T>0$
such that with probability $\eta$, 
action $a^*$ is chosen for all periods after $T$, 
which implies that the limiting frequency of $a^*$ is at least $\eta > 0$.

\begin{proof}[Proof Claim \ref{clm:b1}]
Let $\cF$ be the support of $\pi_0$.
For every pair of  states $F_0, F_1 \in \cF$, let
\begin{equation*}
    X_{a \rightarrow a'}(F_0,F_1) \equiv \log \frac{F_0(y|a')}
    {F_1(y|a')} 
    \textrm{ with probability } F^*(y|a) \textrm{ for every } y \in Y.
\end{equation*}
Let $\bar{l}$ be the largest realization of $|X_{a \rightarrow a'}(F_0,F_1)|$ for any $a,a'\in A$ and $F_0, F_1\in \cF$.
One can verify that the increment of $\log\frac{\pi_t(F_0)}{\pi_t(F_1)}$ follows the same distribution as random variable
$X_{a_{t-k^*} \rightarrow a_{t-k'}}(F_0,F_1)$. Therefore,
\begin{align*}
\log\frac{\pi_T(F_0)}{\pi_T(F_1)} - \log\frac{\pi_0(F_0)}{\pi_0(F_1)} - \sum_{t=1}^T \expect{X_{a_{t-k^*} \rightarrow a_{t-k'}}(F_0,F_1)}
\end{align*}
is a martingale with bounded increments. 
According to Lemma \ref{lem:azuma}, 
for every $\epsilon>0$, $T>0$, $\pi$, and $F_0,F_1$, 
and any sequence of actions $a_0,\dots,a_T$, we have 
\begin{align}\label{eq:b1}
\Pr\left[\log\frac{\pi_T(F_0)}{\pi_T(F_1)} - \log\frac{\pi_0(F_0)}{\pi_0(F_1)} - \sum_{t=1}^T \expect{X_{a_{t-k^*} \rightarrow a_{t-k'}}(F_0,F_1)} \geq \epsilon\right] 
\leq \exp \left(-\frac{\epsilon_1^{2}}{2T\bar{l}^2}\right). 
\end{align}

Now consider an auxiliary scenario where $k^*=k'$.
One can verify that 
$X_{a_{t-k'} \rightarrow a_{t-k'}}(F,F^*) < 0$ for every $F_1\neq F^*$. 
For  every $\epsilon > 0$, 
let $T \in \mathbb{N}$ be large enough such that 
\begin{enumerate}
\item $|\cF| \cdot \exp \left(-\frac{\epsilon_1^{2}}{2T\bar{l}^2}\right) < 1$.
\item For any $F\in \cF\setminus F^*$, 
$- \log\frac{\pi_0(F)}{\pi_0(F^*)} - \sum_{t=1}^T \expect{X{a_{t-k'} \rightarrow a_{t-k'}}(F,F^*)} -\epsilon > \log \frac{|\cF|(1-\epsilon)}{\epsilon}$.
\end{enumerate}
By applying union bound to inequality \eqref{eq:b1} for all $F\in \cF\setminus F^*$, we know that with strictly positive probability, 
$\log\frac{\pi_T(F)}{\pi_T(F^*)} \geq \log \frac{|\cF|(1-\epsilon)}{\epsilon}$ for any $F\in \cF\setminus F^*$. 
This implies that $\pi_T(F^*) \geq 1-\epsilon$. 
Note that by Condition \ref{cd:1}, 
this event also occurs with strictly positive probability when the true state is $F^*$ with time lag $k^*$.
This concludes the proof of Claim \ref{clm:b1}.
\end{proof}

\begin{proof}[Proof of Claim \ref{clm:b2}]
For every finite set $\cF\subseteq \Delta (\mathcal{Y})$ such that $F^* \in \mathcal{F}$,
let $\pi_{\epsilon} \in (0,1)$ be the lowest probability such that for every $\pi \in \Delta (\mathcal{F})$, if $\pi$ attaches probability at least $\pi_{\epsilon}$ to $F^*$, then the agent has a strict incentive to choose $a^*$. 
For simplicity, let $X(F)\equiv X_{a^* \rightarrow a^*}(F,F^*)$ for any $F\in \cF\setminus F^*$.  
By definition, we have $\expect{X(F)} < 0$. 
Let $r^*_F > 0$ be defined via
$\expect[x\sim X(F)]{\exp(r^*_F \cdot x)} = 1$.
According to Lemma \ref{lem:wald}, for a sequence of i.i.d.\ random variables $X_1,\dots,X_t$ that is distributed according to $X(F)$, 
we have 
\begin{align*}
\Pr\left[\bigcup_{n=1}^\infty \left\{\sum_{t=1}^{n} X_t \geq c\right\}\right]
\leq \exp(-r^*_F \cdot c) \textrm{ for every } c>0.
\end{align*}
Let $r^* = \min_{F\in \cF\setminus F^*} r^*_F$.
Let $c$ be such that $|\cF|\cdot \exp(-r^* \cdot c) < 1$, 
and let $\epsilon>0$ be such that $c+\bar{l}\cdot\max\{k^*,k'\}+\log\frac{\epsilon}{1-\epsilon} < \log \frac{\pi_\epsilon}{|\cF|\cdot (1-\pi_\epsilon)}$,
we know that for the first $\max\{k^*,k'\}$ periods, 
the agent always chooses action $a^*$
and the log likelihood ratio between $F$ and $F^*$
increases by at most $\bar{l}\cdot\max\{k^*,k'\}$.
Moreover, for any $t>\bar{l}\cdot\max\{k^*,k'\}$, with probability $1-|\cF|\cdot \exp(-r^* \cdot c)>0$, 
we have that 
$$\log\frac{\pi_t(F)}{\pi_t(F^*)} 
< c+\bar{l}\cdot\max\{k^*,k'\}+\log\frac{\pi_0(F)}{\pi_0(F^*)} 
< c+\bar{l}\cdot\max\{k^*,k'\}+\log\frac{\epsilon}{1-\epsilon} 
< \log \frac{\pi_\epsilon}{|\cF|\cdot (1-\pi_\epsilon)}$$
for any $F\in \cF \setminus F^*$. This implies that $\pi_t(F^*) > \pi_\epsilon$ for any $t>0$
and hence the agent chooses action $a^*$ in all future periods. 
\end{proof}


\section{Proof of Theorem 1}\label{secC}
We establish inequality \eqref{eq:infinitesimal frequency}
when $|Y| = 2$, and later adjust the proof to environments where $|Y| \geq 3$. 
For simplicity, we first consider the  case in which $k^* = 1$
and $k' = 0$. Our argument straightforwardly generalizes to other values of 
$k^*$ and $k'$ as long as $k^* \neq k'$.

Since $|Y|=2$, 
the agent  chooses the action that induces the high-payoff outcome  with higher probability. 
Without loss of generality, let $Y \equiv \{y_0,y_1\}$, and let $v(y_0)=0$ and $v(y_1)=1$.
Let $\zeta, \zeta' \in (0,1/4)$, and let 
\begin{align}\label{eq:construction}
&F^*(y_1 | a=0) \equiv \frac{1}{2} \quad
&F^*(y_1 | a=1) \equiv  \zeta';\nonumber
\\
&F_0(y_1 | a=0) \equiv 2\zeta\quad
&F_0(y_1 | a=1) \equiv \zeta;
\\
&F_1(y_1 | a=0) \equiv 3\zeta \quad
&F_1(y_1 | a=1) \equiv 4\zeta.\nonumber
\end{align}
The optimal action under $F^*$ is $0$.
Let
$\pi_0(F_0) \equiv \frac{1-\zeta'}{2}$, $\pi_0(F_1) \equiv \frac{1-\zeta'}{2}$,
and $\pi(F) \equiv \zeta'$.
First,
we show that for any $\gamma > 0$,
there exist $\zeta$ and $\zeta'$
such that the asymptotic frequency of action $0$ is less than $\gamma$. 
By the end of this section, we identify the crucial components in this construction to show that 
there exists an open set of distributions such that the agent can have arbitrarily small frequency of choosing action $a^*$.

First we bound the expected number of times with which the agent chooses action $0$ and action $1$
when the true distribution is $F^*$
and the agent's prior belief attach small but positive probability to $F^*$. 
Let 
\begin{equation}\label{4.3}
    X_{a \rightarrow a'}(F_0,F_1) \equiv \log \frac{F_0(y|a')}{F_1(y|a')} \textrm{ with probability } F^*(y|a) \textrm{ for every } y \in Y.
\end{equation}
Intuitively, this is the change in the log likelihood ratio between $F_0$ and $F_1$ when the previous period action was $a$ and the current period action is $a'$. 
Let $\overline{l} (F_0,F_1)$ be the largest realization of $X_{1 \rightarrow 0} (F_0,F_1)$, and let
$\underline{l}(F_0,F_1)$ be the smallest realization of $X_{0 \rightarrow 1} (F_0,F_1)$. 
By construction we have $\overline{l}(F_0,F_1)>0$ and $\underline{l}(F_0,F_1)<0$. 
In what follows, we omit the dependence of $(F_0,F_1)$ and write $\overline{l}$ and $\underline{l}$  instead.

Consider a hypothetical scenario in which
the support of the agent's prior belief is $\{F_0,F_1\}$.
For any discount factor $\delta\in [0,1)$,
there exists $l^* \in \mathbb{R}$ depending on $\delta$ such that the agent 
is indifferent between actions $0$ and $1$ when $\log \frac{\pi(F_0)}{\pi(F_1)}=l^*$, 
and strictly prefers action~$0$ if and only if $\log \frac{\pi(F_0)}{\pi(F_1)}\geq l^*$.
For every $l > l^*$,
let random variable $\tau_0 (l)$ be the number of consecutive periods with which the agent takes action $0$ when the initial value of $\log \frac{\pi(F_0)}{\pi(F_1)}$ is $l$.
For every $l<l^*$,
let random variable $\tau_1 (l)$ be the number of consecutive periods with which the agent takes action $1$ when the initial value of $\log \frac{\pi(F_0)}{\pi(F_1)}$ is $l$. 


\begin{Claim}\label{clm:finite}
Random variables $\tau_0(l)$ and $\tau_1(l)$ have finite mean and variance for every $l \in \mathbb{R}$. 
\end{Claim}
\begin{proof}
According to Bayes rule, 
$l_t = l_{t-1} + X_{0 \rightarrow 0}(F_0,F_1)$ for every $t \in \mathbb{N}$.
Let $H$ be the maximal difference in the realization of random variable $X_{0}(F_0,F_1)$. 
The Chernoff-Hoeffding inequality implies that
\begin{align*}
\Pr[l_t > l^*]
\leq \exp\left(-\frac{2(-t\cdot \expect{X_{0 \rightarrow 0}(F_0,F_1)} - l_0+l^*)^2}{tH^2}\right).
\end{align*}
Since $\Pr[l_t > l^*]$ vanishes exponentially as $t \rightarrow +\infty$, $\tau_0(l)$ has finite mean and variance for every $l \geq l^*$. 
Similarly, one can also show that $\tau_1(l)$ has a finite mean and variance for every $l \leq l^*$.
\end{proof}

Next, suppose $F^*$ belongs to the support of agent's prior belief. 
Let $\tau_0^{\epsilon} (l)$ be the number of consecutive periods with which the agent takes action~$0$ when the initial value of $\log \frac{\pi(F_0)}{\pi(F_1)}$ is $l-\epsilon$
and switches to action~$1$ if the log likelihood ratio is below $l^*+\epsilon$
after the first period.
Let $\tau_1^{\epsilon} (l)$ be the number of consecutive periods with which the agent takes action $1$ when the initial value of $\log \frac{\pi(F_0)}{\pi(F_1)}$ is $l+\epsilon$
and switches to action 0 if the log likelihood ratio is above $l^*-\epsilon$
after the first period. 

For every $\epsilon > 0$, there exists $\pi_{\epsilon}>0$ such that when the prior belief $\pi \in \Delta (\cF)$ satisfies $\pi(F^*) < \pi_{\varepsilon}$,
the agent strictly prefers action $0$ when $l > l^* +\varepsilon$ and strictly prefers action $1$ when $l < l^* -\varepsilon$. 
Let $X_{\tau_0}(\epsilon)$ be the random variable that has the same distribution as $\tau_0(\overline{l} + l^*+2\epsilon)$
and 
$X'_{\tau_0}(\epsilon)$ be the random variable that has the same distribution as $\tau^{\epsilon}_0(l^*)$.
Let $X_{\tau_1}(\epsilon)$ be the random variable that has the same distribution as $\tau_1(\underline{l} +l^* - 2\epsilon)$
and 
$X'_{\tau_1}(\epsilon)$ be the random variable that has the same distribution as $\tau^{\epsilon}_1(l^*)$.
For every $\epsilon > 0$ and $\eta > 0$, 
let 
\begin{align*}
\hat{c}_0(\epsilon, \eta) &\equiv \max_{t>0}\left\{(1+\eta) t \bE[X_{\tau_0}(\epsilon)] 
- \ln\bE\left[e^{t X_{\tau_0}(\epsilon)}\right]\right\}\\
c'_0(\epsilon, \eta) &\equiv \max_{t>0}\left\{-(1-\eta) t \bE[X'_{\tau_0}(\epsilon)] 
- \ln\bE\left[e^{-t X'_{\tau_0}(\epsilon)}\right]\right\}\\
c_0(\epsilon, \eta) &\equiv \min\{\hat{c}_0(\epsilon, \eta), c'_0(\epsilon, \eta)\}
\end{align*}
and 
\begin{align*}
\hat{c}_1(\epsilon, \eta) &\equiv \max_{t>0}\left\{(1+\eta) t\bE[X_{\tau_1}(\epsilon)]
- \ln\bE\left[e^{t X_{\tau_1}(\epsilon)}\right]\right\}\\
c'_1(\epsilon, \eta) &\equiv \max_{t>0}\left\{-(1-\eta) t \bE[X'_{\tau_1}(\epsilon)] 
- \ln\bE\left[e^{-t X'_{\tau_1}(\epsilon)}\right]\right\}\\
c_1(\epsilon, \eta) &\equiv \min\{\hat{c}_1(\epsilon, \eta), c'_1(\epsilon, \eta)\}.
\end{align*} 
Let $$K_0(\epsilon, \eta) \equiv \frac{1}{c_0(\epsilon, \eta)} 
\cdot \log \frac{e^{c_0(\epsilon, \eta)}}{\epsilon(e^{c_0(\epsilon, \eta)}-1)}$$
and $$K_1(\epsilon, \eta) \equiv \frac{1}{c_1(\epsilon, \eta)} 
\cdot \log \frac{e^{c_1(\epsilon, \eta)}}{\epsilon(e^{c_1(\epsilon, \eta)}-1)}.$$
Let $\overline{\eta}_0(\epsilon, \eta) \in \mathbb{R}_+$ be such that 
$$c_0(\epsilon, \bar{\eta}_0(\epsilon, \eta)) 
= \frac{\log(K_0(\epsilon, \eta)/\epsilon)}{K_0(\epsilon, \eta)},$$ 
and $\overline{\eta}_1(\epsilon, \eta) \in \mathbb{R}_+$ be such that 
$$c_1(\epsilon, \bar{\eta}_1(\epsilon, \eta)) 
= \frac{\log(K_1(\epsilon, \eta)/\epsilon)}{K_1(\epsilon, \eta)}.$$ For every $k \in \mathbb{N}$, let $t^0_k \in \mathbb{N}$ be the $k$th time such that $a_{t_k} = \az$ and $a_{t_k-1} = \ao$
and $t^1_k \in \mathbb{N}$ be the $k$th time such that $a_{t_k} = \ao$ and $a_{t_k-1} = \az$.
Let $S^0_t$ be the total number of adjacent periods 
until $t$ in which the agent chooses action $\az$.
Let $S^1_t$ be the total number of adjacent periods 
until $t$ in which the agent chooses action $\ao$.

\begin{Claim}\label{clm:concentrate return time}
For every $\epsilon > 0$ and $\eta > 0$, if $\pi(F^*) < \pi_{\varepsilon}$,
then the following event happens with probability at least $1-6\epsilon$: 
\begin{align*}
S^0_{t^0_k} &\leq [k(1+\eta) + K_0(\epsilon, \eta)(1+\bar{\eta}_0(\epsilon, \eta))]
\bE[\tau_0(\bar{l}+l^* + 2\epsilon)] \\
S^0_{t^0_k} &\geq (k - K_0(\epsilon, \eta))\cdot(1-\eta)
\bE[\tau^{\epsilon}_0(l^*)]\\
S^1_{t^1_k} &\leq [k(1+\eta) + K_1(\epsilon, \eta)(1+\bar{\eta}_1(\epsilon, \eta))]
\bE[\tau_1(\underline{l} +l^*- 2\epsilon)]\\
S^1_{t^1_k} &\geq (k - K_1(\epsilon, \eta))\cdot(1-\eta)
\bE[\tau^{\epsilon}_1(l^*)]
\end{align*}
for every $k \in \mathbb{N}$. 
\end{Claim}
\begin{proof}
We establish the upper and lower bounds for $S^0_{t^0_k}$. The ones for 
$S^1_{t^1_k}$ can be derived using a similar argument. 
 Claim \ref{clm:finite} implies that $\bE[X_{\tau_0}(\epsilon)]$ is finite. 
Suppose $\pi$ is such that $\pi(F^*) < \pi_{\varepsilon}$. 
Since $S^0_{t^0_k}$ is first order stochastically dominated by $\sum_{i=1}^k x_i$
with $x_i \sim X_{\tau_0}(\epsilon)$, 
Lemma \ref{lem:concentration} implies that
\begin{align*}
\Pr[S^0_{t^0_k} > k(1+\eta)\cdot \bE[\tau_0(\bar{l} +l^*+ 2\epsilon)]] 
\leq \exp(-k\cdot c_0(\epsilon, \eta)).
\end{align*}
The union bound implies that
\begin{align*}
\Pr\left[\bigcup_{k\geq K_0(\epsilon, \eta)} 
\{S^0_{t^0_k} > k(1+\eta) \cdot \bE[\tau_0(\bar{l} +l^*+ 2\epsilon)]\}\right]
\leq \sum_{k\geq K_0(\epsilon, \eta)} \exp(-k\cdot c_0(\epsilon, \eta)) 
\leq \epsilon.
\end{align*}
Moreover, for every $k < K_0(\epsilon, \eta)$, we have
\begin{align*}
\Pr\left[S^0_{t^0_k} > k(1+\bar{\eta}_0(\epsilon, \eta)) \cdot 
\bE[\tau_0(\bar{l} +l^*+ 2\epsilon)]) \right] 
\leq \frac{\epsilon}{K_0(\epsilon, \eta)}. 
\end{align*}
Take the union of these events, we have
\begin{align*}
&\Pr\left[\bigcup_{k\geq 1} \{S^0_{t^0_k} > k(1+\eta) \cdot 
\bE[\tau_0(\bar{l}+l^* + 2\epsilon)] 
+ K_0(\epsilon, \eta) (1+\bar{\eta}(\epsilon,\eta)) \cdot \bE[\tau_0(\bar{l} + l^*+2\epsilon)]\}\right]\\
&\leq \Pr\left[\bigcup_{k\geq K_0(\epsilon, \eta)} 
\{S^0_{t^0_k} > k(1+\eta) \cdot 
\bE[\tau_0(\bar{l} +l^*+ 2\epsilon)] \}\right]
+ \Pr\left[\bigcup_{k < K_0(\epsilon, \eta)} \{S^0_{t^0_k} > k(1+\bar{\eta}_0(\epsilon, \eta))\}\right]
\leq 2\epsilon.
\end{align*}
Moreover, $S^0_{t^0_k}$ first order stochastically dominates $\sum_{i=1}^k x_i$
with $x_i \sim X'_{\tau_0}(\epsilon)$. 
Lemma \ref{lem:concentration} implies that
\begin{align*}
\Pr[S^0_{t^0_k} < k(1-\eta)\cdot \bE[\tau_1^{\epsilon}(l^*)]] 
\leq \exp(-k\cdot c_0(\epsilon, \eta)).
\end{align*}
By union bound, we have 
\begin{equation*}
\Pr\left[\bigcup_{k\geq K_0(\epsilon, \eta)} 
\{S^0_{t^0_k} < k(1-\eta) \cdot \bE[\tau_1^{\epsilon}(l^*)]\}\right]
\leq \sum_{k\geq K_0(\epsilon, \eta)} \exp(-k\cdot c_0(\epsilon, \eta)) 
\leq \epsilon.\qedhere
\end{equation*}
\end{proof}

\begin{Claim}\label{clm:small prob on G}
For every $\epsilon > 0$, 
there exists a prior belief $\pi_0 \in \Delta (\cF)$ such that event
$\{\pi(F^*) < \pi_{\varepsilon}\}$ 
occurs
with probability at least $1-6\epsilon$.
\end{Claim}
\begin{proof}
Let $l_t(F,F') \equiv \log \frac{\pi_t(F)}{\pi_t(F')}$.
Let $X_t(F^*) \equiv l_t(F^*, F_0) - l_{t-1}(F^*, F_0)$, 
$Z_0(F^*) \equiv l_0(F^*, F_0)$ and $Z_t(F^*) = Z_{t-1}(F^*) + l_t(F^*, F_0) - l_{t-1}(F^*, F_0) - \expect{X_t(F^*) \Big| h^{t-1}}$ 
for every $t \geq 1$. 
One can verify that  $\{Z_t(F^*)\}_{t \in \mathbb{N}}$
is a martingale. 
By definition, $l_t(F^*, F_0) = Z_t(F^*) + \sum_{t'\leq t} \expect{X_{t'}(F^*) \Big| h^{t'-1}}$.
Let $H$ be the difference between the maximal realization of $X_0(F^*,F_0)$ and the minimal realization of $X_0(F^*,F_0)$. 
We have $|Z_{t}(F^*)-Z_{t-1}(F^*)|\leq H$ for every $t \in \mathbb{N}$.
According to Lemma \ref{lem:azuma}, 
\begin{align*}
\Pr[Z_{t}(F^*)-Z_{0}(F^*)\geq t\eta]
\leq \exp\left(-\frac{t\eta^2}{2H^2}\right) \textrm{ for every } \eta \in \mathbb{R}_+.
\end{align*}
Let $T \equiv \frac{2H^2}{\eta^2} \log \frac{e^{\eta^2/(2H^2)}}{\epsilon(e^{\eta^2/(2H^2)}-1)}$.
Take the union of these events, we obtain the following upper bound:
\begin{align*}
\Pr\left[\bigcup_{t\geq T} \{Z_{t}(F^*)-Z_{0}(F^*) \geq t\eta\}\right]
\leq \sum_{t\geq T}\exp\left(-\frac{t\eta^2}{2H^2}\right)
\leq \epsilon.
\end{align*}
Moreover, for every $t < T$, we have 
\begin{align*}
\Pr\left[Z_{t}(F^*)-Z_{0}(F^*) \geq \frac{2H^2\log(T/\epsilon)}{\eta^2}\right]
\leq \frac{\epsilon}{T}.
\end{align*}
Take the union bound, we obtain 
\begin{align}\label{eq:bound on 1}
\Pr\left[\bigcup_{t\geq 1} \left\{l_{t}(F^*, F_0)-l_{0}(F^*, F_0) 
\geq t\eta + \frac{2H^2\log(T/\epsilon)}{\eta^2} 
+ \sum_{t'\leq t}\expect{X_{t'}(F^*) \given h^{t'-1}} \right\}\right]
\leq 2\epsilon.
\end{align}
Similarly, let $X'_t(F^*) \equiv l_t(F^*, F_1) - l_{t-1}(F^*, F_1)$, we have 
\begin{align}\label{eq:bound on 2}
\Pr\left[\bigcup_{t\geq 1} \left\{l_{t}(F^*, F_1)-l_{0}(F^*, F_1) 
\geq t\eta + \frac{2H^2\log(T/\epsilon)}{\eta^2} 
+ \sum_{t'\leq t}\expect{X'_{t'}(F^*) \given h^{t'-1}} \right\}\right]
\leq 2\epsilon.
\end{align}
Let 
\begin{align*}
T'_0 &= K_0(\epsilon_1, \eta)(1+\bar{\eta}_0(\epsilon_1, \eta)) \bE[\tau_0(\bar{l} +l^*+ 2\epsilon_1)] 
+ K_0(\epsilon_1, \eta)\bE[\tau^{\epsilon}_0(l^*)] \\
T'_1 &=K_1(\epsilon_1, \eta)(1+\bar{\eta}_1(\epsilon_1, \eta))
\bE[\tau_1(\underline{l}+l^* -2\epsilon_1)]
+ K_1(\epsilon_1, \eta)\bE[\tau^{\epsilon}_1(l^*)].
\end{align*}
Note that the expected log likelihood of $F^*$
when switching from action $0$ to action $1$ is $\sum_{y \in Y} F^*(y|a) \log F(y|a')$, which
diverges to $-\infty$ as  $\zeta'$ goes to $0$, 
while the the expected log likelihood of $F\in\{F_0,F_1\}$ remains bounded.
Thus, according to Claim \ref{clm:concentrate return time}, 
for sufficiently small $\zeta'$, 
\begin{align*}
\sum_{t'\leq t}\expect{X_{t'}(F^*) \given h^{t'-1}}
\leq T'_0\cdot (\bar{Z}+\eta) - t\cdot \eta \textrm{ for every } t\in \mathbb{N}
\end{align*}
occurs with probability at least $1-2\epsilon_1$. 
Similar bound holds for $\sum_{t'\leq t}\expect{X'_{t'}(F^*) \given h^{t'-1}}$.
Combining these with inequalities \eqref{eq:bound on 1} and \eqref{eq:bound on 2},
for sufficiently small $l_0(F^*, F_0), l_0(F^*, F_1)$, we have
\begin{align*}
\Pr\left[\bigcup_{t\geq 1} \{l_{t}(F^*, F_0) \geq \log \pi_{\epsilon} \text{ or } l_{t}(F^*, F_1) \geq \log \pi_{\epsilon}\}\right]
\leq 6\epsilon.
\end{align*}
Therefore, the probability of the event that $\pi_t(F^*) < \pi_{\epsilon}$ for every $t \in \mathbb{N}$
is at least $1-6\epsilon$. 
\end{proof}
According to Claim \ref{clm:small prob on G}, there exists a positive probability event under which the probability that the agent's posterior belief attaches to $F^*$ is sufficiently small in all periods. 
Conditional on this event, Claim \ref{clm:concentrate return time} implies that the agent's action cycles between $0$ and $1$.
In the last step, we bound the asymptotic frequency of action $0$:
\begin{align*}
&\sup_{\sigma \in \Sigma^*(\pi_0)} \Big\{ \limsup_{t \rightarrow +\infty}    \mathbb{E}^{\sigma} \Big[
\frac{1}{t} \sum_{s=1}^t \mathbf{1}\{a_s=a^*\}
\Big]   \Big\} 
= \limsup_{k \rightarrow +\infty}\  \mathbb{E}\left[\frac{S^0_{t^0_k}}{S^1_{t^1_k}+S^0_{t^0_k}}\right]\\
&\leq \limsup_{k \rightarrow +\infty}
\frac{[k(1+\eta) + K_0(\epsilon, \eta)(1+\bar{\eta}_0(\epsilon, \eta))]
\bE[\tau_0(\bar{l}+l^* + 2\epsilon)]}
{(k - K_1(\epsilon, \eta))\cdot(1-\eta)
\bE[\tau^{\epsilon}_1(l^*)]} + 8\epsilon\\
&= \frac{\bE[\tau_0(\bar{l}+l^* + 2\epsilon)]}
{(1-\eta) \bE[\tau^{\epsilon}_1(l^*)]} 
+ 8\epsilon
\leq \gamma. 
\end{align*}
The first inequality holds by directly applying Claims \ref{clm:concentrate return time} and \ref{clm:small prob on G}
and note that when the events in the claims fails, 
the expected frequency is at most 1.
Moreover, for any $\gamma' < \gamma$, 
there exists $\zeta > 0$ in the construction of $F_0, F_1$
such that $\frac{\bE[\tau_0(\bar{l}+l^* + 2\epsilon)]}
{(1-\eta) \bE[\tau^{\epsilon}_1(l^*)]} < \gamma$
due to the fact that in our construction, 
$\bE[\tau_0(\bar{l}+l^* + 2\epsilon)]$
is bounded from above for any $\zeta > 0$
while $\bE[\tau^{\epsilon}_1(l^*)] \geq \frac{1}{2\log\frac{1-\zeta}{1-4\zeta}}$
approaches infinity as $\zeta\to 0$.
Thus, the last inequality holds by simply setting $\epsilon,\eta,\zeta$ to be small enough constants.

\paragraph{Remark:}
Note that proof of Theorem \ref{Theorem1} does not hinge on the parameters in the design of the instance in \eqref{eq:construction}. 
We summarize the important features of the construction of $F^*$ and $\pi_0$:
\begin{enumerate}
\item The KL-divergence between $F^*(\cdot|0)$ and $F^*(\cdot|1)$ is sufficiently large.\footnote{KL-divergence is not symmetric, and it is sufficient to have either $D(F^*(\cdot|0) \,||\, F^*(\cdot|1))$ or $D(F^*(\cdot|1) \,||\, F^*(\cdot|0))$ is large. } 
This is sufficient to establish that
the posterior probability of $F^*$ converges to $0$. 

\item $F_1(\cdot|0)$ is closer to $F^*(\cdot|0)$ compared to $F_0(\cdot|0)$, 
and $F_0(\cdot|1)$ is closer to $F^*(\cdot|1)$ compared to $F_1(\cdot|1)$. 
This is sufficient to established that the action cycles between $0$ and $1$ for infinite periods if the agent does not believe $F^*$ happens with high probability. 

\item The expected log likelihood ratio $\expect{X_{1 \rightarrow 1}(F_0,F_1)}$ is sufficiently close to $0$. 
This is sufficient to established that the number of time periods required for the agent to switch action from $1$ to $0$ is sufficiently large, 
which implies the limit frequency of action $0$ is sufficiently small. 
\end{enumerate}
As is evident from Claim \ref{clm:concentrate return time} and \ref{clm:small prob on G}, 
essentially any instance satisfying those three properties is sufficient to show that the limit frequency of $a*$ is sufficiently small, 
and the example in \eqref{eq:construction} is an illustration that satisfies all three properties. 
Next we discuss the generalization the inefficiency result to broader settings. 
\begin{itemize}
\item When $|Y| \geq 3$, 
there exists a subset $Y' \subseteq Y$ and $|Y'| = 2$
such that each distribution in the support of the agent's belief 
coincides for outcomes $y\in Y \backslash Y'$.
For any $\gamma > 0$, by setting distributions $F^*, F_0, F_1$ such that 
(1) the probability that the realized outcome $y\in Y \backslash Y'$ is sufficiently small , and
(2) the conditional distribution on $Y'$ is the same as what we constructed for the case there are only two outcomes, 
we can show that the expected average frequency of choosing action $a^*$ can be smaller than $\gamma$.

\item When $|A|\geq 3$, there exists a subset $A' \subseteq A$ and $|A'| = 2$. 
For any $\gamma > 0$, by setting distributions $F^*, F_0, F_1$ such that 
(1) it is always suboptimal to choose any action $a \in A\setminus A'$ for any distribution $\pi$, and
(2) the distribution when choosing action $a\in A'$ is the same as what we constructed for the case there are only two actions, 
we can show that the expected average frequency of choosing action $a^*$ can be smaller than $\gamma$.

\item For general time lag $k^*\neq k'$, we need to have an additional step to show that the number of periods before the agent switches the action is not always $|k^*-k'-1|$.
This is obvious when $k^*=1$ and $k'=0$. 
For the general case, we can show that the probability of such event happens is strictly between $(0,1)$, 
and when the KL-divergence between $F^*(\cdot|0)$ and $F^*(\cdot|1)$ is sufficiently large, 
the attribution error is sufficiently large, 
and the posterior belief on $F^*$ still converges to $0$.
\end{itemize}

\section{Proof of Theorem 2}\label{secD}
Appendix \ref{secD1} 
characterizes the principal's payoff in the auxiliary game.
Appendix \ref{secD2} establishes the connections between the principal's payoff in the auxiliary game and his payoff in the original game with symmetric uncertainty.

\subsection{Payoff in the Auxiliary Game}\label{secD1}
This section examines the principal's asymptotic payoff in an auxiliary game in which he knows the true state but the agent is naive in the sense that she ignores the informational content of the principal's proposals and updates her belief based only on the chosen policies and observed signals. For every $F \in \cF$, let
\begin{equation}\label{D.1}
\underline{u}(\sigma_p, F) \equiv    \liminf_{t \rightarrow +\infty} \frac{1}{t}\mathbb{E}^{\sigma_p}\Big[\sum_{s=1}^{t} a_s \Big| F \Big]
\end{equation}
and
\begin{equation}\label{D.2}
\overline{u}(\sigma_p,F) \equiv   \limsup_{t \rightarrow +\infty} \frac{1}{t}\mathbb{E}^{\sigma_p}\Big[\sum_{s=1}^{t} a_s \Big| F \Big]
\end{equation}
be the lower and upper bounds on the principal's asymptotic payoffs when he uses strategy $\sigma_p$ and the true state is $F$. Let
\begin{equation}\label{D.3}
    \underline{U} (F) \equiv \sup_{\sigma_p \in \Sigma_p} \underline{u}(\sigma_p,F)
   \quad \textrm{and} \quad   \overline{U} (F) \equiv \sup_{\sigma_p \in \Sigma_p} \overline{u}(\sigma_p,F)
\end{equation}
We establish two lemmas. 
\begin{Lemma}\label{LD.1}
For every $\pi_0 \in \Delta (\cF)$ that has full support, we have $\underline{U}(F_1)=\overline{U}(F_1)=1$. 
\end{Lemma}
\begin{proof}
Let $\pi_{t, i}$ be the posterior probability of distribution $F_i$ according to the agent's belief in period $t$.  
Let $l_t \equiv \log \frac{\pi_{t,1}}{\pi_{t,0}}$.
There exist a threshold $l^*$ such that
the agent chooses action $a_t = 0$ if and only if $l_t > l^*$.

We show that for any $\epsilon > 0$,
the following strategy for the principal achieves payoff at least $1-\epsilon$.
The strategy of the principal is to always propose action $0$ until the log likelihood satisfies 
$l_t > l^* + c$, 
where $c$ is defined later in the analysis. 
The principal switches to always proposing action $1$ if the above condition is satisfied. 
Note that when the principal propose action $0$ for all periods, there is no attribution error, 
and the agent learns the correct distribution. 
By inequality \eqref{eq:b1}, for any $\epsilon_1>0$, any prior $\pi_0$, and any parameter $c$, 
there exists $T>0$ such that with probability at least $1-\epsilon_1$, $l_T > l^* + c$.
Thus with probability at least $1-\epsilon_1$, the principal switches to proposing action 1 before time $T$. 
Moreover, by Claim~\ref{clm:b2}, for any $\epsilon_2 > 0$, 
there exists $c > 0$ such that 
with probability at least $1-\epsilon_2$, 
$l_{t+T} > l_{T} - c$ for all $t>0$.
By setting $\epsilon_1 = \epsilon_2 = \frac{\epsilon}{2}$
and apply the union bound, 
with probability at least $1-\epsilon$, 
we have $l_t > l^*$ for any $t>T$. 
Thus the payoff of the principal is at least $1-\epsilon$ with the given strategy. 
Taking $\epsilon \to 1$ gives the desired bound. 
\end{proof}

\begin{Lemma}\label{LD.2}
For every $\pi_0 \in \Delta (\cF)$ that has full support,
\begin{equation}\label{D.4}
    \underline{U}(F_0)=\overline{U}(F_0)= \qstar \lambdaf.
\end{equation}
\end{Lemma}
The proof consists of two parts. In Section \ref{secD.1}, we show that 
\begin{equation}\label{D.6}
    \sup_{\sigma_p \in \Sigma_p} \overline{u} (\sigma_p, F_0) \leq \qstar \lambdaf \textrm{ for every } F_0 \in \cF_0.
\end{equation}
In Section \ref{secD.2}, we show that
\begin{equation}\label{D.5}
    \sup_{\sigma_p \in \Sigma_p} \underline{u} (\sigma_p, F_0) \geq \qstar \lambdaf \textrm{ for every } F_0 \in \cF_0.
\end{equation}
\subsubsection{Proof of Lemma D.2: Establish the Payoff Upper Bound}\label{secD.1}
Let
\begin{equation*}
    \Pi_1 \equiv  \Big\{
    \pi \in \Delta (\cF) \Big|
\argmax_{i \in \{0,1\}}  \big\{ \sum_{F \in \cF} \pi(F) \sum_{y \in Y} F(y|i)v(y)\big\}=\{1\} 
    \Big\}
\end{equation*}
be the set of beliefs under which the agent strictly prefers action $1$. 
Claim \ref{clm:b2} implies that there exists $\underline{p}>0$ such that for every $\pi_t \notin \Pi_1$ and $\sigma_p \in \Sigma_p$, we have
\begin{equation}\label{D.7}
    \Pr \Big(\pi_s \notin \Pi_1 \textrm{ for every }  s \geq t \Big| F_0,\sigma_p\Big) >\underline{p}.
\end{equation}
For every $k \in \mathbb{N}$, we say that $\pi_t$ crosses $\Pi_1$ in period $k$ (or equivalently, there is a \textit{crossing} in period $k$) if $\pi_{k-1} \in \Pi_1$ and $\pi_k \notin \Pi_1$, or $\pi_{k-1} \notin \Pi_1$ and $\pi_k \in \Pi_1$. The uniform lower bound in (\ref{D.7})
implies that for every $\sigma_p \in \Sigma_p$, the expected number of crossings is finite almost surely. Therefore, 
\begin{equation}\label{D.8}
    \Pr \Big(\exists t \in \mathbb{N} \textrm{ s.t. } \pi_s \notin \Pi_1 \textrm{ for every } s \geq t \Big| F_0,\sigma_p\Big) +
    \Pr \Big(\underbrace{\exists t \in \mathbb{N} \textrm{ s.t. } \pi_s \in \Pi_1 \textrm{ for every } s \geq t}_{\equiv\textrm{event } \mathcal{E}_{\sigma_p}} \Big| F_0,\sigma_p\Big)=1
\end{equation}
Let $\mathcal{E}_{\sigma_p}$ be the event that there exists $t \in \mathbb{N}$ such that $\pi_s \in \Pi_1$ for every $s \geq t$.
Let $\sigma_p^*$ be the strategy of the principal for maximizing probability of event $\mathcal{E}_{\sigma^*_p}$
given prior $\pi_0$.
By definition, 
$\qstar$ is the probability of event $\mathcal{E}_{\sigma_p^*}$ when the principal uses strategy $\sigma_p^*$. 

The principal's asymptotic payoff conditional on event $\{\exists t \in \mathbb{N} \textrm{ s.t. } \pi_s \notin \Pi_1 \textrm{ for every } s \geq t\}$ is zero. We conclude the proof by showing that his asymptotic payoff conditional on event 
$\mathcal{E}_{\sigma_p}$
is at most $\lambdaf$ for every $\sigma_p \in \Sigma_p$ satisfying $\Pr[\mathcal{E}_{\sigma_p}| F_0,\sigma_p]>0$.

Suppose toward a contradiction that there exists $\varepsilon>0$ such that conditional on $\mathcal{E}_{\sigma_p}$, the asymptotic frequency of policy $1$ is more than  $\lambdaf+\varepsilon$ when the true state is $F_0$. 
First, we observe that the asymptotic frequency of $(a_{t-1},a_t)=(1,0)$ equals that of $(a_{t-1},a_t)=(0,1)$ regardless of the principal's strategy $\sigma_p$. 
The definition of $\lambdaf$ then implies that  
\begin{equation}\label{D.9}
    \lim_{t \rightarrow +\infty} \mathbb{E} \Big[\log \frac{\pi_t(F_0)}{\pi_t(F_1)} \Big| \mathcal{E}_{\sigma_p}, \sigma_p \Big] =+\infty. 
\end{equation}
As a result, the agent strictly prefers action $0$ asymptotically when the principal uses strategy $\sigma_p$ conditional on event $\mathcal{E}_{\sigma_p}$. 
This contradicts the definition of event $\mathcal{E}_{\sigma_p}$ under which the agent strictly prefers policy $1$.

\subsubsection{Proof of Lemma D.2: Attain the Payoff Upper Bound}\label{secD.2}
We construct $\sigma_p^{\varepsilon} \in \Sigma_p$ for every $\varepsilon>0$ such that 
\begin{equation}\label{D.11}
    \underline{u}(\sigma_p^{\varepsilon},F_0) \geq \qstar \lambdaf-\varepsilon. 
\end{equation}
For every $\sigma_p \in  \Sigma_p$ and $l \in \mathbb{R}$, let $\mathcal{E}_{\sigma_p,l}$ be the following event when the principal uses strategy $\sigma_p$,
\begin{equation}\label{D.12}
    \log \frac{\pi_s (F_1)}{\pi_s (F_0)} \geq l \textrm{ for every } s \geq t \textrm{ for some } t\in \mathbb{N}. 
\end{equation}
If $\Pi(l) \subset \Delta (\cF)$ be the
set of beliefs that satisfy (\ref{D.12}). 
Recall the definition of $\mathcal{E}_{\sigma_p}$ in (\ref{D.8}), which implies the existence of $l^* \in \mathbb{R}_+$ such that $\mathcal{E}_{\sigma_p,l} \subset \mathcal{E}_{\sigma_p}$ 
and $\Pi(l) \subset \Pi_1$
for every $l \geq l^*$. 
Recall that $\sigma_p^*$ is the strategy that maximizes the probability of event $\mathcal{E}_{\sigma^*_p}$
given prior $\pi_0$.. 
\begin{Lemma}\label{LD.3}
For every $\pi_0 \in \Delta (\cF)$ that has full support
and $l \in \mathbb{R}$, 
we have $\Pr[\mathcal{E}_{\sigma_p^*,l}|F_0,\sigma_p^*] = \qstar$.
\end{Lemma}
\begin{proof}[Proof of Lemma D.3:] As shown before, there exists $\underline{p} >0$ such that for every $\pi_t \notin \Pi_1$, the probability with which 
$\pi_s \notin \Pi_1$ for every $s \geq t$ is at least $\underline{p}$ when the true state is $F_0$. 
As a result, for every $l \in \mathbb{R}_+$, the probability of the following event is zero under any strategy in $\Sigma_p$:
\begin{itemize}
    \item $\pi_s \in \Pi_1 \backslash \Pi(l)$ for every $s \geq t$. 
\end{itemize}
This implies that for every $l \in \mathbb{R}$ that satisfies $\Pi(l) \subset \Pi_1$, we have
\begin{equation}\label{C.10}
\Pr \Big(\exists t \in \mathbb{N} \textrm{ s.t. } \pi_t \in \Pi(l) \textrm{ for all } s \geq t \Big| F_0,\sigma_p^*\Big) 
= 
\Pr \Big(\exists t \in \mathbb{N} \textrm{ s.t. } \pi_t \in \Pi_1 \textrm{ for all } s \geq t \Big| F_0,\sigma_p^*\Big). 
\end{equation}
and moreover,
\begin{equation}\label{C.11}
    \Pr \Big(\exists t \in \mathbb{N} \textrm{ s.t. } \pi_t\in \Pi_1 \textrm{ for all } s \geq t \Big| F_0,\sigma_p^*\Big) +
    \Pr \Big(\exists t \in \mathbb{N} \textrm{ s.t. } \pi_t \notin \Pi_1 \textrm{ for all } s \geq t \Big| F_0,\sigma_p^*\Big)=1.
\end{equation}
Equations (\ref{C.10}) and (\ref{C.11}) together imply that 
$\Pr[\event_{\strategy_p^*, l}| F_0,\sigma_p^*]
=\Pr[\event_{\strategy^*_p}| F_0,\sigma_p^*]$, 
while the latter equals $\qstar$. 
\end{proof}

Next we focus on the case when $\lambdaf > 0$
since the case $\lambdaf = 0$ is trivial. 
In this case, we know that 
$\mathbb{E}[X_{1 \rightarrow 0} +X_{0 \rightarrow 1}] > 0$
since 
$\mathbb{E}[X_{1 \rightarrow 1}] < 0$.

For small enough $\varepsilon>0$, let $T_1, T_2 \in \mathbb{N}$ be two positive integers such that $T_1$ is even and 
$\frac{\frac{T_1}{2}+T_2}{T_1+T_2} \in \Big(\lambdaf-\epsilon, \lambdaf\Big)$. 
Let $T \equiv T_1+T_2$. Let $\overline{\strategy}_p \in \Sigma_p$ be defined as: 
\begin{itemize}
    \item $\overline{\sigma}_p (h^t)= 0$ if there exists $k \in \mathbb{N}$ such that $t \in \{kT+2,kT+4,...,kT+T_1\}$,
     \item $\overline{\sigma}_p (h^t)= 1$  otherwise. 
\end{itemize}
According to $\overline{\sigma}_p$, the frequency with which the principal proposes policy $1$ 
belongs to the interval $(\lambdaf - \epsilon, \lambdaf)$.

Let $l_t \equiv \log \frac{\pi_t (F_1)}{\pi_t(F_0)}$,
and let $X_T$ be the increment of $l_t$ from period $t$ to $t+T$ 
when policy $0$ is chosen in period $t+2,t+4,...,t+T_1$ 
and policy $1$ is chosen in other periods. 
Let $\overline{H}$ be the maximal realization of $X_T$.

Let $r^* > 0, \eta > 0$ be  such that 
$\expect[z\sim X_T]{\exp(r^*z)} = 1$
and $\exp(-r^* \cdot \eta) < \epsilon$. 
Let $\bar{l} \in \mathbb{R}$ be large enough such that $\Pi(\bar{l}) \subset \Pi_1$. 
Recall the definition of $\sigma_p^*$.
Let $\strategy^{\varepsilon}_p \in \Sigma_p$ be defined as:
\begin{itemize}
\item $\sigma^{\varepsilon}_p(h^t)=\sigma_p^*(h^t)$ if
$\pi_t \in \Pi(\bar{l}+\eta+\overline{H})$
for all $t'< t$;
\item $\sigma^{\varepsilon}_p(h^t)=\overline{\strategy}_p(h^t)$ otherwise. 
\end{itemize} 
Conditional on $\pi_t$ reaches $\Pi(\bar{l}+\eta+\overline{H})$, the Wald's inequality in Lemma \ref{lem:wald} implies that the probability with which $\pi_s \in \Pi(\bar{l})$ for every $s \geq t$ is at least $1-\varepsilon$, which implies that the principal's asymptotic payoff is at least 
$\qstar (\lambdaf-\varepsilon)$ when the true state is $F_0$.

\subsection{Connections between Auxiliary Game \& Original Game}\label{secD2}
We show that the principal's payoff in the original game equals his expected payoff in the auxiliary game studied in Appendix \ref{secD1}. First, we show that the principal learns the true state asymptotically regardless of the chosen policies.
\begin{Lemma}\label{lem:D4}
For every $\sigma_p \in \Sigma_p$, $F \in \cF$, and $\varepsilon >0$, there exists $\tau \in \mathbb{N}$ such that
\begin{equation}\label{E.1}
    \Pr \Big(\pi_{\tau} (F) > 1-\varepsilon \Big| F \Big) > 1-\varepsilon. 
\end{equation}
\end{Lemma}
\begin{proof}[Proof of Lemma \ref{lem:D4}:]
Let $Q_{F}$ be the probability measure over $\mathcal{H}$ induced by distribution $F$ and let $Q_p$ be the probability measured over $\mathcal{H}$ 
induced by the principal's prior belief $\pi_0 \in \Delta (\cF)$. 
For every history $h^t$, let $q_{F|h^t} \in \Delta (A \times Y)$ be the principal's belief about $(a_t,y_t)$ conditional on the true state being $F$, 
and let $q_{\pi_t|h^t} \in \Delta (A \times Y)$ be the principal's belief about $(a_t,y_t)$ 
when his belief about the state is $\pi_t \in \Delta (\cF)$. 
The chain rule for relative entropy implies that
\begin{equation}\label{E.2}
  -\log \pi_0(F) \geq   d\Big(Q_{F} \Big\| Q_{p} \Big) =
    \sum_{t=0}^{\infty} \mathbb{E}_{Q_{F}} \Big[
    d\Big(
    q_{F|h^{t}}
    \Big\|
    q_{\pi_t|h^{t}}
    \Big)
    \Big].
\end{equation}
Conditional \ref{cd:1} implies that $d(q_{F|h^t} || q_{F' |h^t}) >0$ for every $F \neq F'$. Since $F$ is finite, for every $\varepsilon>0$, 
there exists $\eta>0$ such that 
$d(q_{F|h^{t}}\| q_{p|h^{t}}) > \eta$ 
for every $\pi_0 \in \Delta (\cF)$ satisfying $\pi_0(F) \leq 1-\varepsilon$.
Inequality (\ref{E.2}) implies the existence of $\tau \in \mathbb{N}$ such that 
\begin{equation}\label{E.3}
    \sum_{t=\tau}^{\infty} \mathbb{E}_{Q_{F}} \Big[
    d\Big(
    q_{F|h^{t}}
    \Big\|
    q_{\pi_t|h^{t}}
    \Big)
    \Big] \leq \eta \varepsilon. 
\end{equation}
The Markov's inequality implies that the probability with which $d(q_{F|h^{t}}\|    q_{\pi_t|h^{t}}) > \eta$ is strictly less than $\varepsilon$ for every $t \geq \tau$, or equivalently, the probability with which $\pi_t(F) \leq 1-\varepsilon$ is less than $\varepsilon$ for every $t \geq \tau$. 
\end{proof}
\begin{Lemma}\label{LD.5}
We have $\overline{V}=\underline{V}=\sum_{F \in \cF} \pi_0(F) \overline{U}(F) =\sum_{F \in \cF} \pi_0(F) \underline{U}(F)$. 
\end{Lemma}
\begin{proof}[Proof of Lemma \ref{LD.5}:] 
Since $\overline{U}(F)=\underline{U}(F)$ for every $F \in \cF$, we have $\overline{V} \leq  \sum_{F \in \cF} \pi_0(F) \overline{U}(F)$. We show  $\underline{V} \geq  \sum_{F \in \cF} \pi_0(F) \overline{U}(F)$ by constructing a strategy $\sigma_p^{\varepsilon}$ for every $\varepsilon>0$ under which $\underline{V}(\sigma_p^{\varepsilon}) \geq  \sum_{F \in \cF} \pi_0(F) \overline{U}(F)-\varepsilon$.

For every $\varepsilon >0$, let $\tau \in \mathbb{N}$ be such that  $\Pr \Big(\pi_{\tau} (F) > 1-\varepsilon \Big| F \Big) > 1-\varepsilon$, 
and let $\sigma_p^{F_0}(\varepsilon) \in \Sigma_p$ be the strategy under which 
the principal obtains utility $\underline{u}(\sigma_p^{F_0}(\varepsilon), F_0) \geq  \overline{U} (F_0)-\varepsilon$
if the true state is $F_0$. 
Such a strategy exists according to Lemma \ref{LD.2}.
Similarly, let $\sigma_p^{F_1}(\varepsilon)$
be the strategy under which 
the principal obtains utility $\underline{u}(\sigma_p^{F_1}(\varepsilon), F_1) \geq  1-\varepsilon$.
Let $\sigma_p^{\varepsilon} \in \Sigma_p$ be a strategy under which 
\begin{itemize}
    \item The principal follows $\sigma_p^{F_0}(\varepsilon)$ for every $t \leq \tau$.
    \item If $\pi_{\tau} (F_0) \geq 1-\varepsilon$, then the principal follows $\sigma_p^{F_0}(\varepsilon)$ starting from period $\tau$. 
\item Otherwise, he follows $\sigma_p^{F_1}(\varepsilon)$ starting from period $\tau$. 
\end{itemize}
Next, we establish a lower bound on the principal's asymptotic payoff from strategy $\sigma_p^{\varepsilon}$. Conditional on the true state is $F_0$, 
the probability with which the principal plays $\sigma_p^F(\varepsilon,\tau)$ is at least $1-\varepsilon$.
Conditional on the true state is $F_1$, there exists $T>0$ such that the probability with which the principal proposes $1$ in every period after $\tau+T$ is greater than $1-\varepsilon$ according to Lemma \ref{LD.1}. 
As a result, the principal's asymptotic payoff from $\sigma_p^{\varepsilon}$ is at least 
$(1-\varepsilon)\sum_{F \in \cF} \pi_0(F) \overline{U}(F) $. Since the principal's stage-game payoff is between $0$ and $1$, we have
$(1-\varepsilon)\sum_{F \in \cF} \pi_0(F) \overline{U}(F) \geq \sum_{F \in \cF} \pi_0(F) \overline{U}(F)-\varepsilon$.
\end{proof}

\paragraph{Generalizations.}
Finally, we discuss the generalization of our result in broader settings. 
\begin{enumerate}
\item All results in this section does not hinge on the fact that $k^*=1$ and $k'=0$. 
In fact, all the lemmas and claims hold directly for general time lags. 

\item When $|\cF| > 2$, we denote $\cF_0$ as the set of distributions with optimal action $0$ for the agent 
and $\cF_1$ as the set of distributions with optimal action $1$ for the agent.
The results directly generalize when $|\cF_1| > 1$. 
As we observe from Lemma \ref{LD.1}, the payoff of the principal in the auxiliary game does not depend on the prior when the true state is in $\cF_1$. 
Thus the principal can simply learn the true distribution with high probability as described in Lemma \ref{LD.5}. 
However, things are more complicated when $|\cF_0|>1$. 
The main reason is that the payoff of the principal depends on the prior $\pi_0$ in the auxiliary game when the true state is in $\cF_1$. 
When the principal faces uncertainly over $\cF_0$, 
if there does not exist a strategy $\strategy_p$ that maximizes the probability of the event $\event_{\strategy_p}$ simultaneously for all $F_0\in \cF_0$, 
the principal suffers a non-negligible utility loss in the process of learning the true state. 
\end{enumerate}